\documentclass{article}

\usepackage{a4wide}							
\usepackage{cite}									
\usepackage{amsmath}    			
\usepackage{graphicx}   				
\usepackage{subfigure}					
\usepackage{amssymb}    			
\usepackage{cancel}							
\usepackage{mathtools}					
\mathtoolsset{showonlyrefs=true}							
\usepackage{amsthm}						
\theoremstyle{plain}
\newtheorem{theorem}{Theorem}
\newtheorem{lemma}[theorem]{Lemma}	
	
\newtheorem{corollary}[theorem]{Corollary}	
\theoremstyle{remark}
\newtheorem*{remark}{Remark}

\def\<{\left\langle} 
\def\>{\right\rangle}
\def\l({\left(}				
\def\r){\right)}						
\def\E{\mathbb{E}}			
\def\P{\mathbb{P}}										
\def\R{\mathbb{R}}

\def\Z{\mathbb{Z}}

\def\A{\mathcal{A}}
\def\B{\mathcal{B}}
\def\D{\mathcal{D}}
\def\F{\mathcal{F}}				
\def\H{\mathcal{H}}						
\def\L{\mathcal{L}}			
									
\def\O{\mathcal{O}}
\def\eps{\epsilon}
\def\om{\omega}
\def\sig{\sigma}

\def\Lam{\Lambda}
\def\lam{\lambda}

\def\Wtilde{\widetilde{W}}
\def\Btilde{\widetilde{B}}
\def\Etilde{\widetilde{\mathbb{E}}}						
\def\Ptilde{\widetilde{\mathbb{P}}}			
\def\Bhat{\widehat{B}}
\def\What{\widehat{W}}
\def\Xhat{\widehat{X}}
\def\Yhat{\widehat{Y}}

\def\Y{Y^\eps}

\def\sigmabar{\overline{\sigma}}

\def\d{\partial}		

\begin{document}


\title{Time-Changed Fast Mean-Reverting Stochastic Volatility Models}
\author{Matthew Lorig \thanks{Department of Operations Research \& Financial Engineering, Princeton University, Princeton, NJ  08544, {\em mlorig@princeton.edu}.}}
\date{\today}
\maketitle


\begin{abstract}
We introduce a class of randomly time-changed fast mean-reverting stochastic volatility (TC-FMR-SV) models.  Using spectral theory and singular perturbation techniques, we  derive an approximation for the price of any European option in the TC-FMR-SV setting.  Three examples of random time-changes are provided and are shown to induce distinct implied volatility surfaces.  The key features of the TC-FMR-SV framework are that $(i)$ it is able to incorporate jumps into the price process of the underlying asset $(ii)$ it allows for the leverage effect and $(iii)$ it can accommodate multiple factors of volatility, which operate on different time-scales.
\end{abstract}

\section{Introduction}
Stochastic volatility models have played an important role in the derivatives markets over the past twenty years.  Much of the success of stochastic volatility models is due to the fact that two of the earliest and most well-known models--the Heston model \cite{heston1993} and the Hull-White model \cite{hullwhite1987}--capture the most salient features of the implied volatility surface while simultaneously preserving the analytic tractability needed to quickly calculate the price of an option.  Yet the short-comings of these models is well-documented in literature.  For example, the Heston model misprices far in- and out-of-the-money European options \cite{fiorentini2002, zhang}.
\par
There are a number of possible explanations for why the earliest stochastic volatility models fail to match implied volatility levels across all strikes and maturities.  One theory is that a single factor of volatility, running on a single time scale, is not sufficient for describing the dynamics of the volatility process.  Indeed, the existence of several factors of volatility has been documented in literature \cite{alizadeh2001, anderson, chernov2003, engle, fouque2003, hillebrand, lebaron, melino, muller}.  Such evidence has led to the development of multi-scale stochastic volatility models, i.e. models in which instantaneous volatility levels are controlled by multiple diffusions running of different time-scales \cite{fouque2004multiscale, lorig, perelló2004multiple}.
\par
Another line of reasoning states that jumps in the underlying asset price are required in order to capture the true dynamics of the market.  Empirical work supports this notion  \cite{chernov2003}.  Hence, academics and practitioners have developed models that incorporate both jumps in the asset price as well as stochastic volatility \cite{bates1996jumps, duffiepansingleton, scott1997pricing}.
\par
Along these lines, Mendoza-Arriaga et al. recently introduced a unified credit-equity framework in which the underlying asset is modeled as a stochastically time-changed scalar diffusion \cite{carr}.   This work is notable for a number of reasons.  First, the scalar diffusion that controls the asset price may exhibit both local volatility (i.e. volatility that is a function of the scalar diffusion itself) and killing (i.e. jump to default).  When the local volatility is modeled as a negative power of the scalar diffusion a decrease in the underlying asset price results in an increase in volatility.  This feature, known as the leverage effect, has been empirically documented \cite{bouchaud2001leverage}.  Additionally, by subjecting the scalar diffusion to a random time-change the authors are able to incorporate jumps in the asset price as well as non-local factors of stochastic volatility.  Finally, it is shown that, under relatively benign conditions, the framework of Mendoza-Arriaga et al. remains analytically tractable.
We see great value in the work of Mendoza-Arriaga et al. and seek to build upon it.
\par
In this paper, rather than base our model upon a scalar diffusion as in \cite{carr}, we begin with the class of fast mean-reverting stochastic volatility (FMR-SV) models considered by Fouque et al. in \cite{fouque}.  Such models are important because they capture the empirically known-to-exist short time-scale of volatility \cite{fouque2003, hillebrand2005}.  Additionally, FMR-SV models capture the leverage effect by negatively correlating the Brownian motions that drive the asset price and volatility processes.  Using the methods outlined by Mendoza-Arriaga et al. in \cite{carr}, we subject the FMR-SV class of models to a random time-change.  For certain classes of time-changes this has the effect of adding jumps to the underlying asset price as well as additional factors of volatility.  These additional factors of volatility operate on a different time-scale than the fast mean-reverting factor volatility.  We refer to this class of models as the class of \emph{time-changed} fast mean-reverting stochastic volatility (TC-FMR-SV) models.
\par
The rest of this paper proceeds as follows.  In section \ref{sec:model} we introduce the class of TC-FMR-SV models.  This is done in a few steps.  First, in section \ref{sec:FPS}, we review the class of FMR-SV models considered in \cite{fouque}.  Next, in subsection \ref{sec:TimeChangeFPS}, we explain how the FMR-SV class can be extended using random time-changes.  Finally, in section \ref{sec:TimeChange}, we review the three classes of random time-change.  Some specific model assumptions are listed in section \ref{sec:assumptions}.
\par
In section \ref{sec:pricing} we develop our option-pricing methodology.  Again, this is done in several steps.  First, in section \ref{sec:Spectral} we review some important results from spectral theory, which we immediately apply to the European option-pricing problem in the TC-FMR-SV setting.  This reduces the option-pricing problem to that of solving a single eigenvalue equation.  In section \ref{sec:eigenvalue}, we find an approximate solution to this eigenvalue equation using techniques from singular perturbation theory.  Then, in section \ref{sec:prices} we show how to relate the approximate solution of the eigenvalue equation to the approximate price of a European option.  The main result of our work is the formula we provide in Theorem \ref{thm:main} for the approximate price of a European option in the TC-FMR-SV setting.
\par
In section \ref{sec:accuracy} we prove the accuracy of our option-pricing approximation.  And in section \ref{sec:examples} we provide examples of four different random time-changes and calculate the approximate price of a European call option in these time-change regimes.

\section{Model Framework}	\label{sec:model}
In this section we introduce a class of TC-FMR-SV models.  We begin by reviewing the FMR-SV class of models considered by Fouque et al in \cite{fouque}.

\subsection{Review of FMR-SV Models}	\label{sec:FPS}
Under the physical measure $\P$ the FMR-SV class of models has the following dynamics
\begin{align}
S_t							&=			\exp \left( \beta \, t + X_t \right) ,	\\ 
dX_t						&=			- \frac{1}{2}f^2(Y_t^\eps) dt + f(Y_t^\eps) dW_t	,	&X_0 	&=	x	,	\label{eq:dXphys}\\
dY_t^\eps		&=			\frac{1}{\eps}\left(m - Y_t^\eps\right) dt
										+ \frac{\nu\sqrt{2}}{\sqrt{\eps}}  dB_t	,						&Y_0^\eps 	&=	y	,	\label{eq:dYphys}\\
d\<W,B\>_t	&=			\rho \, dt.	\label{eq:correlation}
\end{align}
Here, $W_t$ and $B_t$ are Brownian motions under $\P$ with instantaneous correlation $\rho \in [-1,1]$.  The process $S_t$ represents the price of a non-dividend paying asset (stock, index, etc.), which has expected geometric growth rate $\beta>0$ and stochastic volatility $f(Y^\eps_t)>0$.  The process $Y_t^\eps$ appears as an Ornstein–-Uhlenbeck (OU) process with long-run mean $m \in \R$ and ``vol of vol'' $\nu>0$.  The OU process operates on time-scale $\eps>0$, which is intended to be small so that the rate of mean-reversion $(1/\eps)$ of $\Y_t$ is high.  It is in this sense that $Y_t^\eps$ is fast mean-reverting.  In fact, $Y_t^\eps$ need not be an OU process specifically.  The essential aspect of $Y_t^\eps$ is that it be an ergodic process with a unique invariant distribution.
The function $f(y)$ is left unspecified, as only certain moments of $f(y)$ play a role in the FMR-SV framework.  Specific assumptions on the function $f(y)$ and the process $\Y_t$ will be given in section \ref{sec:assumptions}.
The parameter $\eps$ will play an important role throughout this paper.  As such, we will use a superscript $\eps$ to indicate dependence on this small time-scale parameter.
\par
For the purpose of option-pricing, it is necessary to move to the risk-neutral pricing measure, which we denote as $\Ptilde$.  Under $\Ptilde$ the FMR-SV class of models has the following dynamics
\begin{align}
S_t							&=			\exp \left( r t + X_t \right)	,																							\label{eq:dS}	\\
dX_t						&=			- \frac{1}{2}f^2(Y_t^\eps) dt + f(Y_t^\eps) d\Wtilde_t	,			&X_0 	&=	x	,	\label{eq:dX}	\\
dY_t^\eps		&=			\left[\frac{1}{\eps}\left(m - Y_t^\eps\right)
										- \frac{\nu \sqrt{2}}{\sqrt{\eps}}\Gamma(Y_t^\eps)\right]dt
										+ \frac{\nu\sqrt{2}}{\sqrt{\eps}}  d\Btilde_t	,								&Y_0	&=	y	,	\label{eq:dY}	\\
d\<\Wtilde,\Btilde\>_t	&=			\rho \, dt.
\end{align}
Here, $\Wtilde_t$ and $\Btilde_t$ are Brownian motions under $\Ptilde$ with instantaneous correlation $\rho$.  The Girsanov transformation, which relates the physical measure $\P$ to the risk-neutral measure $\Ptilde$, is chosen such that 
the volatility-driving process $Y_t^\eps$ acquires a market price of volatility risk $\Gamma(Y_t^\eps)$ and such that the discounted asset price $\l(e^{-rt} S_t\r)$ is a martingale under $\Ptilde$.  Note that $r>0$ is the risk-free rate of interest.  As was the case with $f(y)$, the function $\Gamma(y)$ is left unspecified, as only certain moments of $\Gamma(y)$ play a role in the FMR-SV framework.  Specific assumptions on $\Gamma(y)$ will be given in section \ref{sec:assumptions}.

\subsection{TC-FMR-SV Models}	\label{sec:TimeChangeFPS}
Under the risk-neutral measure $\Ptilde$, the TC-FMR-SV models have the following form
\begin{align}
S_t					&=			\exp \left(r t + X_{T_t} \right)	.		\label{eq:STimeChanged}
\end{align}
Here, $(X_t,\Y_t)$ is as described by equations \eqref{eq:dX} - \eqref{eq:dY} in section \ref{sec:FPS}.  The key difference between the TC-FMR-SV class of models and the FMR-SV class is that the dynamics of the $\log$ of the discounted asset price $\log \left( e^{-rt} S_t \right)$, which would simply be given by the two-dimensional Markov diffusion $(X_t,Y_t^\eps)$ in the FMR-SV framework, is now given by a \emph{time-changed} diffusion $(X_{T_t},Y_{T_t}^\eps)$.  Broadly speaking, a random time-change $T_t$ is an increasing process, starting from zero, which is independent of $(X_t,Y_t^\eps)$ and is right-continuous with left limits.  We list specific assumptions on the random time-change $T_t$ in section \ref{sec:assumptions}.

\subsection{Stochastic Time-Changes}\label{sec:TimeChange}
In this paper, we will consider three classes of stochastic time-changes: L\'{e}vy subordinators, absolutely continuous time-changes, and time-changes that are the composition of a L\'{e}vy subordinator and an absolutely continuous time-change.  These three classes of stochastic time-change are employed extensively in \cite{carr} in the context of local volatility models with state-dependent killing rates.  Drawing inspiration from \cite{carr}, we will use these classes of time-change in the context of FMR-SV models.  A review of each of these classes is presented below.  In an effort to avoid re-inventing the wheel, our discussion will be brief, focusing mainly on those aspects necessary for calculating option prices.  For a more detailed discussion of stochastic time-changes, we refer the reader to \cite{carr}.

\subsubsection{L\'{e}vy Subordinator $T_t^1$} \label{sec:Levy}
A \emph{L\'{e}vy Subordinator} $T_t^1$ is a non-decreasing L\'{e}vy process with positive jumps and non-negative drift.  Because all L\'{e}vy processes have stationary and independent increments, the Laplace transform of a L\'{e}vy subordinator can be expressed as
\begin{align}
\Etilde \left[ e^{-\Lambda T_t^1} \right]		&=		e^{ -\phi(\Lambda) t  },	 \qquad \l(\Lam \in \mathcal{I}\r)	,	\label{eq:Levy}	\\
\mathcal{I}	&:=	\left\{ \Lam \in \R : \Etilde \left[ e^{-\Lambda T_t^1} \right]	< \infty \right\}	.
\end{align}
The function $\phi(\Lambda)$ is known as the \emph{L\'{e}vy exponent} of the subordinator $T^1_t$ and is given by the L\'{e}vy-Khintchine formula
\begin{align}
\phi(\Lambda)																&=		\gamma \Lambda + \int_0^\infty \left(1-e^{-\Lambda s} \right) \nu(ds)	.	\label{eq:LevyKintchine}
\end{align}
Because all L\'{e}vy subordinators are of finite variation no truncation of integral \eqref{eq:LevyKintchine} is necessary.  The absence of a $\Lambda$-independent constant term in \eqref{eq:LevyKintchine} means that we have excluded any killing of the stochastic time-change.  We require that the \emph{drift} $\gamma$ of the subordinator $T_t^1$ be non-negative $\gamma \geq 0$.
\par
The \emph{L\'{e}vy measure} $\nu$, which must satisfy
\begin{align}
\int_0^\infty \left(1 \wedge s\right) \nu(ds)	&< \infty	,
\end{align}
describes the arrival rate and distribution of jumps.  Specifically, for some Borel set $B \in \B(\R_+)$, the value $\nu(B)$ gives the intensity of a Poisson process that counts the number of jumps of size $s \in B$.  
\par
For $\Lam \geq 0$ expectation \eqref{eq:Levy} is always finite.  However, in order to prove the accuracy of our pricing approximation in section \ref{sec:accuracy}, we will need to consider the case $\Lam<0$.
To characterize the set $\mathcal{I}$, we recall Theorem $25.17$ of \cite{sato1999levy}, where it is established that 
\begin{align}
\Etilde \left[ e^{-\Lambda T_t^1} \right]	&< \infty \quad \forall \, t &\Longleftrightarrow& & 	\int_1^\infty e^{-\Lam s} \nu(ds) &< \infty	.
\end{align}
In general, $\mathcal{I}$ is an interval $(\underline{\Lam},\infty)$ or $[\underline{\Lam},\infty)$ where $\underline{\Lam} \leq 0$.
\par
An important sub-class of L\'{e}vy subordinators are the subordinators of compound Poisson type.  The jump component of such subordinators is described by a compound Poisson process with (net) jump arrival intensity $\alpha>0$ and jump size distribution $F$.  For such subordinators, the L\'{e}vy measure $\nu(ds)$ can be written
\begin{align}
\nu(ds)	&=	\alpha F(ds)	,
\end{align}
in which case the L\'{e}vy exponent, given by equation \eqref{eq:LevyKintchine}, becomes
\begin{align}
\phi(\Lambda)		&=		\gamma \Lambda + \alpha \left( 1 - \int_0^\infty e^{\Lambda s} F(ds) \right)	.		\label{eq:LevyKintchine2}
\end{align}
\par
Although it is not strictly necessary for our framework, for the sake of computational simplicity, we will be primarily interested in L\'{e}vy subordinators for which the L\'{e}vy exponent $\phi(\Lambda)$ is known in closed form.
\par
We would like to emphasize the importance of L\'{e}vy subordinators as a class of stochastic time-change.  Because L\'{e}vy subordinators $T_t^1$ exhibit jumps, the time-changed diffusion $(X_{T_t^1},Y_{T_t^1}^\eps)$ (and thus the asset price $S_t$) will exhibit jumps as well.  To our knowledge, this is the first time that jumps in the asset price $S_t$ have been incorporated into the FMR-SV framework.
\par
This concludes our brief review of L\'{e}vy subordinators.  For more thorough coverage, we refer the reader to \cite{bertoin}.\\

\subsubsection{Absolutely Continuous Time-Change $T_t^2$}
We now consider stochastic time-changes of the absolutely continuous type.  When we say $T_t^2$ is an \emph{absolutely continuous time-change} we mean that $T_t^2$ can be written as
\begin{align}
T_t^2		&=		\int_0^t V(Z_s) ds	,	&		Z_0	&= z	,	\label{eq:T2}
\end{align}
where $Z_t$ is an infinite lifetime Markov process taking values in $\R^d$.  The function $V:\R^d \rightarrow \R_+$ shall be referred to as the \emph{rate function} of stochastic time-change.  
We are primarily interested in absolutely continuous time-changes $T_t^2$ for which the Laplace transform 
\begin{align}
L(t,z,\Lambda) 	&=		\Etilde_z \left[ e^{-\Lambda T_t^2} \right]	,	 \label{eq:Laplace}		\\
												&=		\Etilde_z \left[ e^{-\Lambda \int_0^t V(Z_s) ds} \right]	,
												&			\l( \Lambda \in \mathcal{J}_t \r)	,		\\
\mathcal{J}_t				&:=		\left\{ \Lam \in \R : \Etilde_z \left[ e^{-\Lambda T_t^2} \right]	< \infty \right\}	.
\end{align}
is known explicitly.  Here, the notation $\Etilde_z\left[ \cdot \right]$ is used to indicate the conditional expectation $\Etilde\left[ \cdot | Z_0 = z \right]$.  We note that \eqref{eq:Laplace} is always finite for $\Lambda \geq 0$.  Values of $\Lam < 0$ for which \eqref{eq:Laplace} is finite depend on the specific choice of $T_t^2$ and must be checked on a case-by-case basis.
\par
In the previous section, we showed that L\'{e}vy subordinators are an important class of stochastic time-change because jumps in the subordinator induce jumps in the asset price.  Absolutely continuous time-changes are important for a very different reason; they have the ability to change a one-factor stochastic volatility model into a multi-factor stochastic volatility model.  To see this, we define 
$\left(\Xhat_t,\Yhat_t^\eps\right) := \left(X_{T_t^2},Y_{T_t^2}^\eps\right)$.
Then, there exist $\Ptilde$-Brownian motions $\What_t$ and $\Bhat_t$ with correlation $\rho$ such that \cite{oksendal}
\begin{align}
d\Xhat_t						&=			- \frac{1}{2}f^2(\Yhat_t^\eps) V(Z_t)dt + f(\Yhat_t^\eps) \sqrt{V(Z_t)} d\What_t	,	&
\Xhat_0							&=			x , \\
d\Yhat_t^\eps		&=			\left[\frac{1}{\eps}\left(m - \Yhat_t^\eps\right)
												- \frac{\nu \sqrt{2}}{\sqrt{\eps}}\Gamma(\Yhat_t^\eps)\right]V(Z_t) dt
												+ \frac{\nu\sqrt{2}}{\sqrt{\eps}}  \sqrt{V(Z_t)} d\Bhat_t	, &
\Yhat_0^\eps			&=			y .
\end{align}
Note that the volatility of $\Xhat_t$ is controlled by the product $f(\Yhat_t^\eps) \sqrt{V(Z_t)}$ rather than just the single factor $f(Y_t^\eps)$, which controls the volatility of $X_t$.
Note also that the multiple factors of volatility are operating on different time-scales; $f(\Yhat_t^\eps)$ acts on a time-scale of $\O(\eps)$ and $\sqrt{V(Z_t)}$ acts on a time-scale of $\O(1)$.  As demonstrated in \cite{lorig}, when compared to their one-factor counterparts, multi-factor stochastic volatility models in which the factors of volatility operate on different time-scales have the ability to vastly improve the fit to the empirically-observed implied volatility surface.

\subsubsection{Composite Time-change $T_t^3$}
Finally, we may consider \emph{composite time-changes}, which are time-changes of the form
\begin{align}
T_t^3							&=		T_{T_t^2}^1	.
\end{align}
Here, $T_t^1$ is a L\'{e}vy subordinator and $T_t^2$ is an absolutely continuous time-change, which is independent of $T_t^1$.  As long as the L\'{e}vy exponent $\phi(\Lambda)$ of $T_t^1$ and the Laplace transform $L(t,z,\Lambda)$ of $T_t^2$ are known explicitly, the Laplace transform of the composite time-change $T_t^3$ can be calculated as well.  This is accomplished by conditioning on the absolutely continuous time-change $T_t^2$ as follows
\begin{align}
\Etilde_z \left[ e^{-\Lambda T_t^3}	\right]
										&=		\Etilde_z \left[ \Etilde_z \left[ e^{-\Lambda T_{T_t^2}^1} \big| \, T_t^2 \right] \right]		\\
										&=		\Etilde_z \left[ e^{-\phi(\Lambda) T_t^2} \right]	\\
										&=		L(t,z,\phi(\Lambda))	,
										&			\l( \Lambda \in \mathcal{K}_t \r)	,													\label{eq:TtLaplace}	\\
\mathcal{K}_t		&:=		\left\{ \Lam \in \R : \Etilde_z \left[ e^{-\Lambda T_t^3} \right]	< \infty \right\}	.	
\end{align}
Again, we remark that $[0,\infty) \subset \mathcal{K}_t$.  However, values of $\Lam < 0$ for which $\Lam \in \mathcal{K}_t$ depend on the specific choice of time-change $T_t^3$.
\par
The importance of composite time-changes is as follows: by combining a L\'{e}vy subordinator with an absolutely continuous time-change we are able to incorporate jumps in the asset price $S_t$ as well as add multiple factors of volatility to the class of FMR-SV models.  The variety gained by combining different types of stochastic time-changes provides us with considerable modeling flexibility.

\subsubsection*{Remark on Notation}
Throughout this paper we shall use the superscripts $1,2,3$ to specify which type of random time-change we wish to consider.  The notation $T_t^1$ will be used to denote a L\'{e}vy subordinator, the notation $T_t^2$ will be used to denote an absolutely continuous time-change and the notation $T_t^3$ will be used to denote a composite time-change.  Finally, if we do not wish to specify a particular class of random time-change we will omit the superscript altogether and use the notation $T_t$.

\subsection{Specific Model Assumptions}\label{sec:assumptions}
We have now described the TC-FMR-SV class of models.  However, we have not been specific about certain technical assumptions on the process $\Y_t$, the functions $f(y)$ and $\Gamma(y)$ and the stochastic time-change $T_t$.  The purpose of this section is to list these assumptions in one place.  Our assumptions are as follows:
\begin{enumerate}
	\item  Under the physical measure $\P$, the process $\Y_t$ is ergodic and has a unique invariant distribution $F_Y$, which is independent of $\eps$ .  We note that this implies the moments $\E [ |Y_t^\eps |^k ]$ are uniformly bounded in $t$.  That is, for every $k \in \Z^+$ there exists a positive constant $C_k < \infty$ such that
	\begin{align}
	\sup_t \E_y [ |Y_t^\eps|^k ]		& \leq 	C_k	.
	\end{align} 	\label{assume:moments}
	\item The volatility function $f(y)$ is a strictly positive function such that
	\begin{enumerate}
	\item the process $\l(X_t,Y_t^\eps\r)$ as the strong solution to SDE's \eqref{eq:dXphys} and \eqref{eq:dYphys} exists and is unique under $\P$,
	\item $\int f^2(y) F_Y(dy) < \infty$,
	\item a solution $\Phi(y)$ to the Poisson equation \eqref{eq:PoissonPhi} exists and is at most polynomially growing.
	\label{assume:fPoisson}
	\end{enumerate}
	\item The function $\Gamma(y)$, which describes the market price of volatility risk, is such that
	\begin{enumerate}
	\item the process $\l( X_t, Y_t^\eps \r)$ as the strong solution of SDE's \eqref{eq:dX} and \eqref{eq:dY} exists and is unique under $\Ptilde$,
	\item there exists a finite constant $C_\Gamma>0$ such that $|\Gamma(y)| \leq C_\Gamma$. \label{item:Cgamma}
	\end{enumerate}
	\item The random time-change $T_t$ is a strictly increasing c\`{a}dl\`{a}g process, which is independent of $\l(X_t,\Y_t \r)$ and satisfies $T_0 =0$ and
	\begin{align}
	C_T		&:=		\Etilde_z \left[ T_t e^{T_t C_\Gamma^2} \right] < \infty	.
	\end{align}	\label{assume:T}
\end{enumerate}

\subsection{The Martingale Condition}\label{sec:martingale}
Although we specified the class of TC-FMR-SV models under a supposedly risk-neutral measure $\Ptilde$, we have not yet shown that the non-dividend-paying asset in these models satisfies the martingale condition
\begin{align}
\Etilde \left[ e^{-rt_2}S_{t_2} \big| \F_{t_1} \right]	 	&=		e^{-rt_1}S_{t_1}	,		& \l( t_1 < t_2 \r)			,
\end{align}
which is required in order for $\Ptilde$ to actually be risk-neutral.  In fact, because $\left(X_{T_t^1},Y_{T_t^1}^\eps\right)$, $\left(X_{T_t^2},Y_{T_t^2}^\eps,Z_t\right)$ and $\left(X_{T_t^3},Y_{T_t^3}^\eps,Z_t\right)$ are time-homogeneous Markov processes, as rigorously established in \cite{carr}, the martingale condition reduces to
\begin{align}
\Etilde_{x,y,z} \left[ e^{-rt}S_{t} \right]	= S_0 = e^{x}	,	\label{eq:martingale}
\end{align}  
where we have used the short-hand notation $\Etilde_{x,y,z}[\cdot]$ to denote the conditional expectation $\Etilde[\cdot|X_0=x,Y_0^\eps=y,Z_0=z]$.  We can verify equation \eqref{eq:martingale} by conditioning on the random time-change $T_t$ as follows
\begin{align}
\Etilde_{x,y,z} \left[ e^{-rt}S_{t} \right]		
		&=	\Etilde_{x,y,z} \left[	\exp \left( x - \frac{1}{2} \int_0^{T_t} f^2(Y_t^\eps) dt + \int_0^{T_t} f(Y_t^\eps) d\Wtilde_t \right) \right]	\\
		&=	\Etilde_{x,y,z} \left[	\Etilde_{x,y,z} \left[ 
												\exp \left( x - \frac{1}{2} \int_0^{T_t} f^2(Y_t^\eps) dt + \int_0^{T_t} f(Y_t^\eps) d\Wtilde_t \right)
												 \Big| \, T_t \right]	\right]	\\
		&=	e^x	,
\end{align}
where we have used the fact that $e^{-\frac{1}{2} \int_0^{T} f^2(Y_t^\eps) dt + \int_0^{T} f(Y_t^\eps) d\Wtilde_t}$ is an exponential martingale.
Having established that the discounted stock-price process $\l(e^{-rt}S_t\r)$ is a martingale under $\Ptilde$ in the TC-FMR-SV framework, we now move on to the option-pricing problem.

\section{Option Pricing}\label{sec:pricing}
In this section, we discuss how the approximate price of any European option can be calculated in the TC-FMR-SV setting.

\subsection{Spectral Representation of European Option Prices}\label{sec:Spectral}
We begin with a brief review of some important results from spectral theory and semigroup operators.
\begin{theorem}\label{thm:spectral}
Suppose $\L$ is a self-adjoint operator acting on a Hilbert space $\H$.  Consider the eigenvalue equation for $(-\L)$
\begin{align}
-\L \, \psi_\lam &= \lam \, \psi_\lam	.
\end{align}
 We denote by $E$ the projection-valued spectral measure of $(-\L)$ and by $\sigmabar(-\L)$ the spectrum of $(-\L)$.  Then:
\begin{enumerate}
\item The operator $(-\L)$ has the following spectral representation
\begin{align}
-\L				&=				 \int_{\sigmabar(-\L)}  \lam \, E \l( d \lam	\r)	.			\label{eq:Lspectral}
\end{align}
If $g$ is a real-valued Borel function on $\R$, then $g(\L)$ can be defined via operational calculus and is given by
\begin{align}
g(-\L)		&=				\int_{\sigmabar(-\L)} g\l( \lam\r) E \l( d\lam \r)	.	\label{eq:gLspectral}
\end{align}
\label{item:spectral}
\item If there exists a number $\gamma> - \infty$ such that $\lam > \gamma$ for all $\lam \in \sigmabar(-\L)$, then
\begin{align}
Q_t &= e^{-t(- \L)}	,						&\l( 0 \leq t < \infty \r),
\end{align}
defines a strongly continuous one-parameter semigroup and the solution to the Cauchy problem
\begin{align}
\l(-\d_t + \L \r) u		&=		0			,		&
u(0,x)									&=		u_0(x)			,			& \l( u_0 \in \H \r)	,
\end{align}
is given by
\begin{align}
u(t,x)		&=		Q_t u_0(x)		=	\int_{\sigmabar(-\L)} e^{- \lam \, t} E\l( d\lam \r) u_0(x) . \label{eq:ThmSolution}
\end{align}
\label{item:semigroup}
\end{enumerate}
\end{theorem}
\begin{proof}
Items \ref{item:spectral} and \ref{item:semigroup} are classical results from functional analysis.  The proof of item \ref{item:spectral} 
can be found in chapter $8$ of \cite{reedsimon} and the proof of item \ref{item:semigroup} is given in Chapter $13$ of \cite{rudin1973functional}.
\end{proof}
\noindent
For convenience, we will write \eqref{eq:ThmSolution} as
\begin{align}
\int_{\sigmabar(\L)} e^{- \lam \, t} E\l( d\lam \r) u_0(x) &= \int_{\sigmabar(\L)} e^{- \lam \, t} \psi_\lam(x) \mu_{u_0}(d\lam) , \label{eq:spec}
\end{align}
 where $\mu_{u_0}(d\lam) = \l(d\psi_\lam,u_0\r)$.  This will help to make it clear that $E\l( d\lam \r) u_0(x)$ is a projection of $u_0(x)$ onto the eigenspace $\H_\lam := \left\{ \psi \in \H : -\L \psi = \lam \psi \right\}$.  In the special case when the spectrum $\sigmabar(-\L)$ is purely absolutely continuous with respect to the Lebesgue measure, equation \eqref{eq:spec} can be written
\begin{align}
\int_\R C_\om \, e^{- \lam_\om \, t} \, \psi_\om(x) \, d\om ,
\end{align}
where the $C_\om$ are constants chosen such that $\int C_\om \, \psi_\om(x) \, d\om =u_0(x)$.
\par
Now, consider a function $u^\eps(t,x,y)$, defined as
\begin{align}
u^\eps(t,x,y)							&:=				\Etilde \left[ h(X_t) \big| X_0 = x, Y_0^\eps = y \right]		, \label{eq:uDef}
\end{align}
The backward variables $x, y$ satisfy the Kolmogorov backward equation
\begin{align}
\left( -\d_t + \L^\eps_{X,Y} \right) \, u^\eps 	&=	0	,	\label{eq:uPDE} \\
u^\eps(0,x,y)																&=	h(x) .	\label{eq:uBC}
\end{align}
Note that the $\d_t$ term carries a minus sign because $t$ is a forward variable.  We use the notation $\L^\eps_{X,Y}$ to indicate the infinitesimal generator of the Markov process $(X_t,Y_t^\eps)$, defined in \eqref{eq:dX} - \eqref{eq:dY}.  For clarity, we write $\L_{X,Y}^\eps$ explicitly and state its domain $\D\l( \L_{X,Y}^\eps \r)$
\begin{align}
\L_{X,Y}^\eps								&=								\frac{1}{\eps} \l( \left(m-y\right)\d_y + \nu^2 \d^2_{yy} \r)
																											+	\frac{1}{\sqrt{\eps}}\l(  \rho \nu \sqrt{2} f(y) \d^2_{xy}-\nu\sqrt{2}\,\Gamma(y)\d_y \r) \\
																	&\qquad				+  \l( -\frac{1}{2}f^2(y)\,\d_x + \frac{1}{2}f^2(y)\,\d^2_{xx}	 \r) , \\
\D\l( \L_{X,Y}^\eps \r)		&=								\left\{ g : \R^2 \rightarrow \R \text{ s.t. } \lim_{t \searrow 0} \frac{\Etilde_{x,y}[ g(X_t,Y_t) ] - g(x,y) }{t} \, \text{exists for all } (x,y) \in \R^2 \right\}	.
\end{align}
Now, suppose we have the solution to the following \emph{eigenvalue equation}
\begin{align}
0		&=		\L_{X,Y}^\eps \Psi_\Lam^\eps		+		\Lambda^\eps \Psi_\Lam^\eps	.	\label{eq:PsiPDE}
\end{align}
By, \emph{solution} to the eigenvalue equation we mean that we have the full set of \emph{eigenvalues} $\left\{\Lambda^\eps\right\}$ and corresponding \emph{eigenfunctions} $\left\{\Psi_\Lam^\eps(x,y)\right\}$ for which \eqref{eq:PsiPDE} holds.  Then \emph{if} the operator $\L_{X,Y}^\eps$ were self-adjoint on some Hilbert space $\H$, by Theorem \ref{thm:spectral}, the solution to Cauchy problem \eqref{eq:uPDE} - \eqref{eq:uBC} could be expressed as
\begin{align}
u^\eps(t,x,y)			&=		\int e^{-\Lambda^\eps t} \Psi_\Lam^\eps(x,y) \, \mu_h \l( d\Lam^\eps	\r),	
\label{eq:uSpectral}
\end{align}
where the measure $\mu_h$ would be determined by BC \eqref{eq:uBC}.
\par
It is not apparent that there exists a Hilbert space $\H$ on which $\L_{X,Y}^\eps$ is self-adjoint.
As such, it is not clear at this point that eigenvalue equation \eqref{eq:PsiPDE} has a solution, nor is it clear that $u^\eps(t,x,y)$ has a representation of the form \eqref{eq:uSpectral}.  Nevertheless, in this paper we do not endeavor to solve the full the eigenvalue equation \eqref{eq:PsiPDE}.  Rather, we shall use techniques from singular perturbation theory to find an approximate solution to \eqref{eq:PsiPDE}.  We will show that the operator associated with the lowest order solution to \eqref{eq:PsiPDE} is in fact self-adjoint on some Hilbert space.  As a result, $u^\eps(t,x,y)$ can be \emph{approximated} by a function of the form \eqref{eq:uSpectral}.  For the moment, however, it will ease our calculations if we assume that $u^\eps(t,x,y)$ can be written as \eqref{eq:uSpectral}, which we shall refer to as the \emph{spectral representation} of $u^\eps(t,x,y)$.
\par
Supposing $u^\eps(t,x,y)$ can be approximated by a function of the form \eqref{eq:uSpectral}, we would like to use this knowledge to find a spectral representation for the price of a European option in the TC-FMR-SV framework.
To this end, we consider a European option with payoff $h(S_t)$ at maturity date $t<\infty$.  Using risk-neutral pricing, and the Markov property of $\left(X_{T_t},Y_{T_t}^\eps,Z_t\right)$, we may write the price of a European option $P^\eps(t,x,y,z)$ as
\begin{align}
P^\eps(t,x,y,z)	&=	e^{-rt}	\Etilde_{x,y,z} \left[ h\left(e^{rt+X_{T_t}}\right) \right]	.
\end{align}
Conditioning on the random time-change $T_t$ we find
\begin{align}
P^\eps(t,x,y,z)		&=	e^{-rt}	\Etilde_{x,y,z} \left[ \Etilde_{x,y,z} \left[ h\left(e^{rt+X_{T_t}}\right) \Big| T_t \right] \right]
												=	e^{-rt}	\Etilde_{x,y,z} \left[ u^\eps (T_t,x,y;t) \right]	.	\label{eq:Ptemp}
\end{align}
Note that $t$ is just a parameter here--not a variable of $u^\eps(T,x,y;t)$.  Now, we use \eqref{eq:uSpectral} to replace $u^\eps(T,x,y;t)$ with its spectral representation.  We have
\begin{align}
P^\eps(t,x,y,z) 		&=		e^{-rt}	\Etilde_{x,y,z} \left[\int e^{-\Lambda^\eps T_t} \Psi_\Lam^\eps (x,y) \mu_h\l( d\Lam^\eps ; t \r)	 \right]	\\
													&=		e^{-rt}	\int	\Etilde_z \left[e^{-\Lambda^\eps T_t} \right] \Psi_\Lam^\eps (x,y) \mu_h\l( d\Lam^\eps ; t \r)		,
													\label{eq:PepsSpectral}
\end{align}
where passing the expectation through the integral is allowed by Fubini's theorem.
We have used the notation $\mu_h\l( d\Lam^\eps ; t \r)$ to remind us that $\mu_h( d\Lam^\eps ; t)$ depends on the paramter $t$ through the BC $u^\eps(0,x,y;t)		=			h(e^{rt + x}) $.
Assuming it exists, we refer to \eqref{eq:PepsSpectral} as the \emph{spectral representation of the option price} $P^\eps(t,x,y,z)$.  Note that the expectation $\Etilde_z [ e^{-\Lambda^\eps T_t}]$ is given explicitly by either \eqref{eq:Levy}, \eqref{eq:Laplace} or \eqref{eq:TtLaplace}, depending on the type of random time-change.  Hence, in order to fully specify the price of the option $P^\eps(t,x,y,z)$, what remains is to solve eigenvalue equation \eqref{eq:PsiPDE} and determine the measure $ \mu_h\l( d\Lam^\eps ; t \r)$ from the BC 
\begin{align}
u^\eps(0,x,y;t)		&=			h(e^{rt + x}) . \label{eq:newBC}
\end{align}

\subsection{Asymptotic Analysis of the Eigenvalue Equation} \label{sec:eigenvalue}
For general $f(y)$ and $\Gamma(y)$ there is no analytic solution to the eigenvalue equation $0 = \L_{X,Y}^\eps\Psi_\Lam^\eps + \Lam^\eps \Psi_\Lam^\eps$.  However, we note that $\L_{X,Y}^\eps$ can be conveniently decomposed in powers of $\sqrt{\eps}$ as
\begin{align}
\L_{X,Y}^\eps						&=		\frac{1}{\eps} \L^{(-2)} + \frac{1}{\sqrt{\eps}} \L^{(-1)} + \L^{(0)}	,	\\
\L^{(-2)}										&=		\left(m-y\right)\d_y + \nu^2 \d^2_{yy}	,		\\
\L^{(-1)}										&=		\rho \nu \sqrt{2} f(y) \d^2_{xy}-\nu\sqrt{2}\, \Gamma(y)\d_y	,	\label{eq:L-1}	\\
\L^{(0)}										&=		-\frac{1}{2}f^2(y)\d_x + \frac{1}{2}f^2(y)\d^2_{xx}	.	
\end{align}
This decomposition suggests a singular perturbative approach.  To this end, we expand $\Psi_\Lam^\eps$ and $\Lam^\eps$ in powers of $\sqrt{\eps}$.  We have
\begin{align}
\Psi_\Lam^\eps
	&=	\Psi_\Lam^{(0)} + \sqrt{\eps} \, \Psi_\Lam^{(1)} + \eps \, \Psi_\Lam^{(2)} + \ldots,		\label{eq:PsiExpansion} \\	
\Lam^\eps
	&=	\Lam^{(0)} + \sqrt{\eps} \, \Lam^{(1)} + \eps \, \Lam^{(2)} + \ldots	\label{eq:LambdaExpansion} .
\end{align}
Expanding in powers of $\sqrt{\eps}$ (rather than some other power of $\eps$) is a natural choice given the form of $\L_{X,Y}^\eps$.  The validity of this expansion will be justified in section \ref{sec:accuracy}, when we establish the accuracy of our pricing approximation.
\par
We now insert the expansions for $\Psi_\Lam^\eps(x,y)$ and $\Lam^\eps$ into eigenvalue equation \eqref{eq:PsiPDE} and collect terms of like-powers of $\sqrt{\eps}$.  The $\O(\eps^{-1})$ and $\O(\eps^{-1/2})$ equations are
\begin{align}
&\O(\eps^{-1}): &
0
	&=	\L^{(-2)} \Psi_\Lam^{(0)} , \\
&\O(\eps^{-1/2}): &
0
	&=	\L^{(-2)} \Psi_\Lam^{(1)} + \L^{(-1)} \Psi_\Lam^{(0)} .
\end{align}
Noting that all terms in $\L^{(-2)}$ and $\L^{(-1)}$ take derivatives with respect to $y$, we may (and do) choose solutions of the form $\Psi_\Lam^{(0)}=\Psi_\Lam^{(0)}(x)$ and $\Psi_\Lam^{(1)}=\Psi_\Lam^{(1)}(x)$ (i.e. functions of $x$ only).  Continuing the asymptotic analysis, the order $\O(\eps^{0})$ and $\O(\eps^{1/2})$ equations are
\begin{align}
&\O(\eps^0): &
0	&=	\L^{(-2)}\Psi_\Lam^{(2)} + \l(\L^{(0)} + \Lam^{(0)} \r) \Psi_\Lam^{(0)} , \label{eq:PoissonPsi2} \\
&\O(\eps^{1/2}): &
0	&=	\L^{(-2)}\Psi_\Lam^{(3)}	+ \L^{(-1)}\Psi_\Lam^{(2)} + \l( \L^{(0)} + \Lam^{(0)} \r)\Psi_\Lam^{(1)} + \Lam^{(1)} \Psi_\Lam^{(0)}	, \label{eq:PoissonPsi3}
\end{align}
where we have used $\L^{(-1)}\Psi_\Lam^{(1)}(x)=0$ in \eqref{eq:PoissonPsi2}.  Equations \eqref{eq:PoissonPsi2} and \eqref{eq:PoissonPsi3}, respectively, are Poisson equations for $\Psi_\Lam^{(2)}(x,y)$ and $\Psi_\Lam^{(3)}(x,y)$ in the variable $y$ of the form
\begin{align}
0 	&= 	\L_Y^1 \Psi \, + \, g \label{eq:PoissonGeneric}
\end{align}
where $\L^{1}_Y = \L^{(-2)}$ is the infinitesimal generator of $Y_t^{1}$ under the physical measure $\P$.  We wish to consider only those solutions $\Psi(y)$ of \eqref{eq:PoissonGeneric} that exhibit at most polynomial growth as $y\rightarrow \pm \infty$.  With this restriction, a necessary condition for the solvability of \eqref{eq:PoissonGeneric} is
\begin{align}
\< g \>
	&:= \int g(y) \, F_Y(dy)
		=	0	.		\label{eq:centering}
\end{align}
We remind the reader that $F_Y$ is the invariant distribution of $Y_t^\eps$ under the physical measure $\P$.  Equation \eqref{eq:centering} is referred to as the \emph{centering condition}.  Please refer to Appendix \ref{sec:Poisson} for a treatise on the Poisson equation and the centering condition.  Throughout this paper, the notation $\< \cdot \>$ will always indicate averaging with respect to the invariant distribution $F_Y$.  In equations \eqref{eq:PoissonPsi2} and \eqref{eq:PoissonPsi3} the centering conditions become
\begin{align}
0
	&=	\l( \< \L^{(0)} \> + \Lam^{(0)} \r) \Psi_\Lam^{(0)} , \label{eq:center1} \\
0
	&=	\<\L^{(-1)}\Psi_\Lam^{(2)}\> + \l( \< \L^{(0)} \> + \Lam^{(0)} \r)\Psi_\Lam^{(1)} + \Lam^{(1)} \Psi_\Lam^{(0)} . \label{eq:center2}
\end{align}
Eigenvalue equation \eqref{eq:center1} can be solved explicitly, as the operator $\< \L^{(0)} \>$ is given by
\begin{align}
\<\L^{(0)}\>
	&= \frac{\sig^2}{2}\l(\d^2_{xx} - \d^2_x \r) ,	&
\sig^2
	&:=	\< f^2 \> .
\end{align}
However, in order to solve equation \eqref{eq:center2}, we need an expression for $\<\L^{(-1)}\Psi_\Lam^{(2)}(x,\cdot)\>$.  To this end, we note from \eqref{eq:PoissonPsi2}
\begin{align}
\L^{(-2)}\Psi_\Lam^{(2)}
	&= 	- \l( \L^{(0)} + \Lam^{(0)} \r) \Psi_\Lam^{(0)}	
	=		- \l( \L^{(0)} - \< \L^{(0)} \> \r) \Psi_\Lam^{(0)}
	=		- \frac{1}{2}\l( f^2 - \sig^2 \r)\l(\d^2_{xx} - \d_x \r)\Psi_\Lam^{(0)}.
\end{align}
Now, introducing $\Phi(y)$ as a solution to the following Poisson equation
\footnote{We note that \eqref{eq:PoissonPhi} satisfies the centering condition and $\Phi(y)$ exists by assumption \ref{assume:fPoisson} of section \ref{sec:assumptions}.}
\begin{align}
\L^{(-2)}\Phi
	&=	f^2 - \sig^2 , \label{eq:PoissonPhi}
\end{align}
we may express $\Psi_\Lam^{(2)}(x,y)$ as
\begin{align}
\Psi_\Lam^{(2)}(x,y)
	&=	- \frac{1}{2}\Phi(y)\l(\d^2_{xx} - \d_x \r)\Psi_\Lam^{(0)}(x) + C(x)	,
\end{align}
where $C(x)$ is some function which is independent of $y$.  Hence, using \eqref{eq:L-1} we find that $\< \L^{(-1)}\Psi_\Lam^{(2)}(x,\cdot) \>$ is given by
\begin{align}
\< \L^{(-1)}\Psi_\Lam^{(2)} (x,\cdot) \>
	&=	\<
			\l( \rho \nu \sqrt{2} f(\cdot) \d^2_{xy}-\nu \sqrt{2} \, \Gamma(\cdot)\d_y \r)
			\l( - \frac{1}{2}\Phi(\cdot)\l(\d^2_{xx} - \d_x \r)\Psi_\Lam^{(0)}(x) + C(x) \r)
			\>	\\
	&= 	\A^{(1)} \, \Psi_\Lam^{(0)}(x)	,
\end{align}
where
\begin{align}
\A^{(1)}
	&=  V_3 \l( \d^3_{xxx} - \d^2_{xx} \r) + V_2 \l( \d^2_{xx}-\d_x \r) , &
V_2
	&= \frac{\nu}{\sqrt{2}} \<\Gamma \d_y \Phi \>,	 &
V_3
	&= \frac{- \rho \nu}{\sqrt{2}} \<f \d_y \Phi \> .	\label{eq:A}
\end{align}
Thus, from \eqref{eq:center2} we have
\begin{align}
0	&=	 \A^{(1)} \, \Psi_\Lam^{(0)}  + \l( \< \L^{(0)} \> + \Lam^{(0)} \r)\Psi_\Lam^{(1)} + \Lam^{(1)} \Psi_\Lam^{(0)} . \label{eq:eigen1}
\end{align}
Given a solution to \eqref{eq:center1}, one can use \eqref{eq:eigen1} to find expressions for $\Psi_\Lam^{(1)}(x)$ and $\Lam^{(1)}$.
\par
This concludes our asymptotic analysis of eigenvalue equation \eqref{eq:PsiPDE}.  Before we present an explicit solution to \eqref{eq:center1} and \eqref{eq:eigen1}, we recall the following result from Sturm-Liouville theory.
\begin{theorem} \label{thm:SturmLiouville}
The eigenfunctions $\Psi_\Lam^{(0)}(x)$ of equation \eqref{eq:center1} form a complete orthonormal basis in the Hilbert space $\H:=L^2(\R,s(x)dx)$ where
\begin{align}
s(x) \, dx
	&= e^{-x} dx , &
\l( u,v \r)_s
	&=	\int \overline{u(x)} \, v(x) \, s(x) \, dx	.
\end{align}
The notation $\overline{u(x)}$ indicates the complex conjugate of $u(x)$.
\end{theorem}
\begin{proof}
The proof is by showing that $\<\L^{(0)}\>$ of equation \eqref{eq:center1} is self-adjoint in $L^2(\R,s(x)dx)$.  This is a standard result of Sturm-Liouville theory.  Details can be found in any number of texts on differential equations \cite{al2008sturm, amrein2005sturm, stakgold2000boundary, zill2008differential, hinton1997spectral}.  
\end{proof}
\begin{theorem}\label{thm:eigenfunctions}
The order $\O\l(\eps^0\r)$ eigenfunctions $\Psi_\om^{(0)}	(x)$ and eigenvalues $\Lambda_\om^{(0)}	$ are given by
\begin{align}
\Psi_\om^{(0)}(x)					&=	\frac{1}{\sqrt{2 \pi}}e^{\l( i \om + 1/2 \r)x} ,	\label{eq:Psi0}	\\
\Lambda_\om^{(0)}				&=	\frac{\sigma^2}{2} \l( \om^2 + \frac{1}{4}\r)	,	\label{eq:Lambda0}
\end{align}
where $\om \in \R$.  The order $\O\l(\eps^{1/2}\r)$ corrections $\Psi_\om^{(1)}(x)	$ and $\Lambda_\om^{(1)}	$ are
\begin{align}
\Psi_\om^{(1)}(x)					&=	0	,	\label{eq:Psi1}	\\ 
\Lambda_\om^{(1)}				&=		-	V_3 \l( \l(i \om + \frac{1}{2}\r)^3 - \l(i \om + \frac{1}{2}\r)^2 \r) - V_2 \l( \l(i \om+ \frac{1}{2}\r)^2 - \l(i \om + \frac{1}{2}\r) \r) ,		\label{eq:Lambda1} 
\end{align}
where $V_2$ and $V_3$ are defined in \eqref{eq:A}.  We note that $\Lam_\om^{(0)} \geq \Lam_{min}^{(0)} := \sig^2/8$.
\end{theorem}
\begin{proof}
A direct substitution shows that \eqref{eq:Psi0}, \eqref{eq:Lambda0}, \eqref{eq:Psi1} and \eqref{eq:Lambda1} satisfy equations \eqref{eq:center1} and \eqref{eq:eigen1}.  One can easily verify that the $\O\l(\eps^0\r)$ eigenfunctions $\left\{ \Psi_\om^{(0)}(x) \right\}$ form a complete basis in $L^2(\R,s(x)dx)$ and satisfy the orthogonality condition
\begin{align}
\l( \Psi_\nu^{(0)} , \Psi_\om^{(0)} \r)_s = \delta(\nu - \om) .	\label{eq:orthogonal}
\end{align}
\end{proof}

\subsection{Option Prices}\label{sec:prices}
We have found explicit expressions for the approximate eigenvalues $\Lam_\om^\eps \approx \Lam_\om^{(0)} + \sqrt{\eps} \, \Lam_\om^{(1)}$ and approximate eigenfunctions $\Psi_\om^\eps(x,y) \approx \Psi_\om^{(0)}(x) + \sqrt{\eps} \, \Psi_\om^{(1)}(x)$.  We now use these expressions to specify the approximate price $P^\eps(t,x,y,z) \approx P^{(0)}(t,x,z) + \sqrt{\eps} \, P^{(1)}(t,x,z)$ of an option.  The following Theorem serves as the main result of our work.
\begin{theorem} \label{thm:main}
The approximate price of an option is given by
\begin{align}
P^\eps(t,x,y,z) 	&\approx 	P^{(0)}(t,x,z) + \sqrt{\eps} \, P^{(1)}(t,x,z) , \\
P^{(0)}(t,x,z) 		&=	e^{-rt} \int C_\om^{(0)}(t) \, \Etilde_z \left[ e^{-\Lam_\om^{(0)} T_t}\right] \Psi_\om^{(0)}(x) d\om	,	\label{eq:P0} \\
P^{(1)}	(t,x,z) 	&=	e^{-rt} \int C_\om^{(0)}(t) \, \Etilde_z \left[ \l( - \Lam_\om^{(1)} T_t \r)e^{-\Lam_\om^{(0)} T_t}\right] \Psi_\om^{(0)}(x) d\om , \label{eq:P1}
\end{align}
where the coefficients $C_\om^{(0)}(t)$ are given by
\begin{align}
C_\om^{(0)}(t)
	&=	\l( \Psi_\om^{(0)}(\cdot), h(e^{rt \, + \, \cdot \,})\r)_s	.		\label{eq:C0}
\end{align}
For a composite time-change $T_t^3$ we have
\begin{align}
\Etilde_z \left[ e^{-\Lam_\om^{(0)} T_t^3}\right]																								&=	L\l(t,z,\phi(\Lam_\om^{(0)})\r) , \label{eq:Expectation0} \\
\Etilde_z \left[ \l( - \Lam_\om^{(1)} T_t^3 \r)e^{-\Lam_\om^{(0)} T_t^3}\right]						&=	\d_\alpha L\l(t,z,\phi_\om^{(1)}\alpha\r)	\Big|_{\alpha=\phi_\om^{(0)}/\phi_\om^{(1)}}	, \label{eq:Expectation1}\\
\phi_\om^{(0)}		&=		\phi(\Lambda_\om^{(0)})	,	\label{eq:phi0} \\
\phi_\om^{(1)}		&=		\d_\alpha \phi(\Lambda_\om^{(1)}\alpha)	\Big|_{\alpha=\Lambda_\om^{(0)}/\Lambda_\om^{(1)}}	.	\label{eq:phi1}
\end{align}
The corresponding expressions for a L\'{e}vy subordinator $T_t^1$ and an absolutely continuous time-change $T_t^2$ can be recovered by setting $L(t,z,\phi)=e^{-\phi \, t}$ and $\phi(\Lam)=\Lam$ respectively.
The $\O \l( \eps^0\r)$ eigenfunctions $\Psi_\om^{(0)}(x)$ and the approximate eigenvalues $\Lam_\om^\eps \approx \Lam_\om^{(0)} + \sqrt{\eps} \, \Lam_\om^{(1)}$ are given in Theorem \ref{thm:eigenfunctions}.
\end{theorem}
\begin{proof}
Consider the spectral representation of $u^\eps(T,x,y;t)$ given by \eqref{eq:uSpectral}.  
Recall that we use the notation $\mu_h \l( d\Lam_\om^\eps	; t \r)$ to remind us that $u^\eps(T,x,y;t)$ has a BC $u^\eps(0,x,y;t)=h(e^{rt+x})$ that takes $t$ as a parameter.
We expand $\mu_h \l( d\Lam_\om^\eps	; t \r)$ and $e^{-\Lam_\om^\eps T}$ in powers of $\sqrt{\eps}$
\begin{align}
\mu_h \l( d\Lam_\om^\eps	; t \r)
	&= C_\om^{(0)} (t) \, d\om+ \sqrt{\eps} \, C_\om^{(1)}(t) \, d\om + \ldots , \label{eq:muExpand}\\
e^{-\Lam_\om^\eps T}
	&= e^{-\Lam_\om^{(0)} T} + \sqrt{\eps} \, \l( - \Lam_\om^{(1)} T \r)e^{-\Lam_\om^{(0)} T}  + \ldots . \label{eq:ExpExpand}
\end{align}
Note that we have expanded the measure $\mu_h \l( d\Lam_\om^\eps ; t \r)$ in terms of a density $C_\om^\eps(t)$ as the spectrum of the $\O\l(\eps^0\r)$ eigenvalue problem is absolutely continuous with respect to the Lebesgue measure $d \om$.  Inserting expansions \eqref{eq:PsiExpansion}, \eqref{eq:muExpand} and \eqref{eq:ExpExpand} into \eqref{eq:uSpectral} and collecting terms of like-powers of $\sqrt{\eps}$ yields
\begin{align}
&\O\l(\eps^0\r):			&
u^{(0)}	(T,x;t)						&=	\int C_\om^{(0)}(t) \, e^{-\Lam_\om^{(0)} T} \Psi_\om^{(0)}(x) d\om,\\
&\O\l(\eps^{1/2}\r):	&
u^{(1)}	(T,x;t)						&=	\int \l( C_\om^{(1)}(t) \, e^{-\Lam_\om^{(0)} T} \Psi_\om^{(0)}(x) + C_\om^{(0)}(t) \, \l(-\Lam_\om^{(1)}T\r)e^{-\Lam_\om^{(0)} T} \Psi_\om^{(0)}(x) \r) d\om,
\end{align}
where we have dropped the $C_\om^{(0)}(t) \, e^{-\Lam_\om^{(0)} T} \Psi_\om^{(1)}(x)$ term because $\Psi_\om^{(1)}(x)=0 $.  Expressions for $C_\om^{(0)}(t)$ and $C_\om^{(1)}(t)$ can be obtained from the BC's $u^{(0)}(0,x;t)=h(e^{rt + x})$ and $u^{(1)}(0,x;t)=0$.  We have
\begin{align}
&\O\l(\eps^0\r):			&
u^{(0)}(0,x;t)
	&=	h(e^{rt+x}) = \int C_\om^{(0)}(t) \Psi_\om^{(0)}(x) d\om	,	\\
&\O\l(\eps^{1/2}\r):	&
u^{(1)}(0,x;t)
	&=	0	= \int C_\om^{(1)}(t) \Psi_\om^{(0)}(x) d\om	.
\end{align}
Hence
\begin{align}
&\O\l(\eps^0\r):	&
\l( \Psi_\nu^{(0)} (\cdot),h(e^{rt \, + \, \cdot \,}) \r)_s
	&=	\int C_\om^{(0)}(t) \l( \Psi_\nu^{(0)} , \Psi_\om^{(0)}  \r)_s d\om = C_\nu^{(0)}(t) ,	\\
&\O\l(\eps^{1/2}\r):	&
0
	&=	\int C_\om^{(1)}(t) \l( \Psi_\nu^{(0)} , \Psi_\om^{(0)} \r)_s d\om = C_\nu^{(1)}(t) ,
\end{align}
where we have used \eqref{eq:orthogonal}.
\par
We have now obtained an explicit expression for $u^\eps(T,x,y;t) \approx u^{(0)}(T,x;t) + \sqrt{\eps}\, u^{(1)}(T,x;t)$.  In order to find an expression for the approximate price of an option $P^\eps(t,x,y,z) \approx P^{(0)}(t,x,z) + \sqrt{\eps}\, P^{(1)}(t,x,z)$ we simply insert our expansion for $u^\eps(T_t,x,y;t)$ into \eqref{eq:Ptemp}, which yields \eqref{eq:P0} and \eqref{eq:P1}.  Expressions \eqref{eq:Expectation0} and \eqref{eq:Expectation1} are given for a composite time-change $T_t^3$and can be obtained by expanding $L\l(t, z, \phi(\Lam_\om^\eps) \r)$ in powers of $\sqrt{\eps}$.
\end{proof}
\begin{remark}
We note that the existence and finiteness of expectations \eqref{eq:Expectation0} and \eqref{eq:Expectation1} is guaranteed for all $\om \in \R$ by assumption \ref{assume:T} of section \ref{sec:assumptions} and by the fact that $\forall \, \om \in \R$ we have $-\Lam_\om^{(0)} \leq - \Lam_{min}^{(0)} < 0 < C_\Gamma^2$.
\end{remark}
\begin{corollary} \label{cor:Vs}
The function $\sqrt{\eps}\,P^{(1)}(t,x,z)$ is linear in the group parameters
\begin{align}
V_2^\eps
	&:= \sqrt{\eps}\,\frac{\nu}{\sqrt{2}} \<\Gamma \d_y \Phi\> = \sqrt{\eps}\,V_2,	 &
V_3^\eps
	&:= -\sqrt{\eps}\,\frac{ \rho \nu}{\sqrt{2}} \<f \d_y \Phi\> =  \sqrt{\eps}\,V_3. \label{eq:Veps}
\end{align}
\end{corollary}
\begin{proof}
First, we note that $V_2$ and $V_3$ do not appear in $C_\om^{(0)}(t)$, $\Lam_\om^{(0)}$ or $\Psi_\om^{(0)}(x)$.  Next, from \eqref{eq:P1} we see that $P^{(1)}(t,x,z)$ is linear in $\Lam_\om^{(1)}$, which, from \eqref{eq:Lambda1}, is a linear function of $V_2$ and $V_3$.  Corollary \ref{cor:Vs} follows immediately.
\end{proof}
Corollary \ref{cor:Vs} relates to a very important feature of the TC-FMR-SV pricing methodology.  Consider first the FMR-SV framework (no time-change).  To describe a particular model within the FMR-SV class, one would have to specify an ergodic diffusion $Y^\eps_t$, a market price of volatility risk $\Gamma(y)$ and a volatility function $f(y)$.  For the purposes of illustration, we chose to specify $Y^\eps_t$ as an OU process.  This choice led us to introduce five unobservable parameters ($m$, $\eps$, $\nu$, $\rho$, $y$).  Note however, that neither the value of these parameters nor the precise form of $\Gamma(y)$ and $f(y)$ are required in order to calculate the approximate price of an option
(the approximate price of an option in the FMR-SV framework is given by setting $T_t = t$ in the TC-FMR-SV framework).
Rather, to $\O(\sqrt{\eps})$, the approximate price of an option can be expressed in terms of $(\sigma^2, V_2^\eps, V_3^\eps)$ as well as the observable parameters $(t,x,r)$.
\par
As mentioned previously, the particular choice of $Y_t^\eps$ as an OU is not central to our analysis.  We could have simply written $Y_t^\eps$ under the physical measure $\P$ as
\begin{align}
dY_t^\eps				&=			\frac{1}{\eps} \alpha(Y_t^\eps) dt + \frac{1}{\sqrt{\eps}} \beta(Y_t^\eps) \, dB_t , & Y_0^\eps = y ,
\end{align}
where $\alpha(y)$ and $\beta(y)$ are such that the assumptions of section \ref{sec:assumptions} are satisfied.
In this case, the group parameters would have become
\begin{align}
V_3^\eps &=	-\sqrt{\eps}\frac{\rho}{2} \< \beta f \d_y \Phi \>, &V_2^\eps &= \sqrt{\eps} \frac{1}{2} \< \beta \Gamma \d_y \Phi \>.
\end{align}
The key point is that, when the volatility-driving process $Y_t^\eps$ is fast mean-reverting and satisfies the conditions of section \ref{sec:assumptions}, the details of the process are unimportant.  In terms of option-pricing, to $\O(\sqrt{\eps})$, all that matters are the values of $(\sigma^2, V_2^\eps, V_3^\eps)$.  This is true regardless of the particular choice of $Y_t^\eps$.
\par
In the TC-FMR-SV framework the situation remains the same -- to calculate the approximate price of an option, precise knowledge of the volatility-driving process $Y_t^\eps$ is \emph{not} required.  However, the particular choice of random time-change $T_t$ \emph{does} affect the approximate price $P^{(0)}(t,x,z) + \sqrt{\eps} \, P^{(1)}(t,x,z)$ of an option.  Thus, when calibrating a particular model within the TC-FMR-SV class to fit market data (be the data quoted option prices or implied volatilities), the unobservable parameters that must be extracted are $(\sig^2, V_2^\eps, V_3^\eps)$ as well as the parameters of the random time-change $T_t$.  We will show in section \ref{sec:examples} that different time-changes induce distinct implied volatility surfaces.  Whether the introduction of time-change parameters is justified by the modeling flexibility the random time-change provides is a topic left for future research.

\section{Accuracy of the Approximation $P^\eps \approx P^{(0)} + \sqrt{\eps}\,P^{(1)}$}\label{sec:accuracy}
In the previous section, we gave a derivation of the approximate price of a European option $P^\eps \approx P^{(0)} + \sqrt{\eps}\,P^{(1)}$ using singular perturbative arguments.  The purpose of this section is to establish the accuracy of this approximation.  In addition to the assumptions listed in section \ref{sec:assumptions}, we shall need one additional assumption for our accuracy proof.
\begin{itemize}
	\item The payoff function $h(e^{rt+x})$ and all derivatives taken with respect to $x$ are smooth and bounded.
\end{itemize}
Obviously, the most common options -- calls and puts -- do not fit this assumption.  To prove the accuracy of our pricing approximation for calls and puts would require regularizing the option payoff as was
done for the class of FMR-SV models in \cite{fouque2003proof}.  The regularization procedure is beyond the scope of this paper.  As such, we limit our analysis to options with smooth and bounded payoffs.
%
\par
Before stating our main accuracy result we need the following Lemma.
\begin{lemma}\label{lem:Jbound}
Suppose $J(y)$ is at most polynomially growing.  Then, there exists a constant $C$ such that
\begin{align}
\Etilde_y \left[ J\left(Y_t^\eps\right) \right] &\leq C e^{t C_\Gamma^2}.
\end{align}
\end{lemma}
\begin{proof}
First, we define $M_t$, the exponential martingale used in Girsanov's theorem to transform the measure on $Y_t^\eps$ from $\P$ to $\Ptilde$
\begin{align}
M_t		&:=		\exp \left(-\int_0^t \Gamma(Y_s^\eps)dB_s - \frac{1}{2}\int_0^t \Gamma^2(Y_s^\eps)ds \right) .
\end{align}
We note
\begin{align}
\E_y \left[ M_t^2 \right]	
	&=			\E_y \left[ \exp \left(	-\int_0^t \left(2\Gamma(Y_s^\eps)\right)dB_s - \frac{1}{2}\int_0^t \left(2\Gamma(Y_s^\eps)\right)^2 ds 
																+ \int_0^t \left(\Gamma^2(Y_s^\eps)\right)ds \right) \right]	\\
	&\leq		\E_y \left[ \exp \left(	-\int_0^t \left(2\Gamma(Y_s^\eps)\right)dB_s - \frac{1}{2}\int_0^t \left(2 \Gamma(Y_s^\eps)\right)^2 ds 
																+ \int_0^t C_\Gamma^2 ds \right) \right]	\\
	&=			e^{ t C_\Gamma^2 }
					\E_y \left[ \exp \left(-\int_0^t 2\Gamma(Y_s^\eps)dB_s - \frac{1}{2}\int_0^t \left(2\Gamma(Y_s^\eps)\right)^2 ds \right) \right] 	\\
	&=			e^{ t C_\Gamma^2 }	,
\end{align}
where we have used assumption \ref{item:Cgamma} of section \ref{sec:assumptions} to bound $\Gamma^2 \l( \Y_t \r)$ by $C_\Gamma^2$.  Hence
\begin{align}
\Etilde_y \left[ |Y_t^\eps|^k \right]		&=		\E_y\left[ |Y_t^\eps|^k M_t\right]
																				\leq	\Big( \E_y \left[ |Y_t^\eps|^{2k} \right] \E_y \left[ M_t^2 \right] \Big)^{1/2}
																				\leq 	\Big( C_{2k} e^{ t C_\Gamma^2 } \Big)^{1/2}	
																				=			\sqrt{C_{2k}} e^{ t C_\Gamma^2/2 }	.
\end{align}
The first inequality is an application of Cauchy-Schwarz.  The second inequality follows from the above bound as well as assumption \ref{assume:moments} of section \ref{sec:assumptions}.  Since $J(y)$ is bounded by a polynomial, this proves lemma \ref{lem:Jbound}.
\end{proof}
From here, we shall proceed as follows.  First, we shall establish the accuracy of the approximation $u^\eps(T,x,y;t) \approx u^{(0)}(T,x;t) + \sqrt{\eps}\, u^{(1)}(T,x;t)$.  Then, we show how this result can be related to the accuracy of the approximate option price $P^\eps(t,x,y,z) \approx P^{(0)}(t,x,z) + \sqrt{\eps}\, P^{(1)}(t,x,z)$.
\begin{theorem}\label{thm:uBound}
For fixed $(T,x,y,t)$ there exists a constant $C$ such that for any $\eps<1$ the solution $u^\eps(T,x,y;t)$ to PDE \eqref{eq:uPDE} with BC \eqref{eq:newBC} satisfies
\begin{align}
\left| u^\eps - \left(u^{(0)}+\sqrt{\eps}\,u^{(1)}\right) \right| &\leq \eps \, C \, T e^{T C_\Gamma^2}	.
\end{align}
\end{theorem}
\begin{proof}
First, we define a remainder term $R^\eps(T,x,y;t)$
\begin{align}
R^\eps	&:=		u^\eps - \left(u^{(0)} + \sqrt{\eps} \, u^{(1)} + \eps \, u^{(2)} + \eps \sqrt{\eps}\,u^{(3)} \right)	.
\end{align}
Next, we see that
\begin{align}
\left(-\d_T + \L_{X,Y}^\eps \right) R^\eps
	&=					\left(-\d_T + \L_{X,Y}^\eps \right) u^\eps - 
							\frac{1}{\eps}\L^{(-2)}u^{(0)} - \frac{1}{\sqrt{\eps}}\left(\L^{(-2)} u^{(1)} + \L^{(-1)} u^{(0)}\right) \\
	&\qquad	 -	\left(\L^{(-2)} u^{(2)} + \L^{(-1)} u^{(1)} + \left(-\d_T+ \L^{(0)}\right) u^{(0)}\right)	\\	
	&\qquad	 -	\sqrt{\eps}\left(\L^{(-2)} u^{(3)} + \L^{(-1)} u^{(2)} + \left(-\d_T+ \L^{(0)}\right) u^{(1)}\right)	\\
	&\qquad	 -	\eps \left(\L^{(-1)} u^{(3)} + \left(-\d_T+ \L^{(0)}\right) u^{(2)}+\sqrt{\eps}\left(-\d_T+ \L^{(0)}\right) u^{(3)}\right)	,	\\
\left(-\d_T + \L_{X,Y}^\eps \right) R^\eps
	&=			- \, \eps \, F^\eps	,	\label{eq:RPDE}\\
F^\eps
	&:=			\left(\L^{(-1)} u^{(3)} + \left(-\d_T+ \L^{(0)}\right) u^{(2)}+\sqrt{\eps}\left(-\d_T+ \L^{(0)}\right) u^{(3)}\right)	,	\\
R^{\eps}(0,x,y;t)
	&=			\eps \, G^\eps(x,y;t)	,	\label{eq:RBC}\\
G^\eps(x,y;t)
	&:=			 - \left( u^{(2)}(0,x,y;t) + \sqrt{\eps}\,u^{(3)}(0,x,y;t) \right)	.
\end{align}
Now, from the Feynman-Kac formula we note that $R^{\eps}(T,x,y;t)$, which is the solution to PDE \eqref{eq:RPDE} with BC \eqref{eq:RBC}, has the following stochastic representation:
\begin{align}
R^{\eps}(T,x,y;t)
	&=	\eps \, \Etilde_{x,y} \left[ G^\eps\left(X_T,Y_T^\eps;t\right) + \int_0^T F^\eps\left(s,X_s,Y_s^\eps;t\right) ds \right]	.
\end{align}
As established in \cite{fouque}, from the boundedness of the payoff function $h\l(e^{rt+x}\r)$, and from the assumptions of section \ref{sec:assumptions}, one can deduce that $F^\eps\left(s,x,y;t\right)$ and $G^\eps\left(x,y;t\right)$ are bounded in $x$ and at most polynomially growing in $y$.  Hence, by lemma \ref{lem:Jbound}, there exists a constant $C_1$ such that
\begin{align}
\left| R^{\eps}(T,x,y;t) \right|		&\leq	\eps \, C_1 \, T e^{T C_\Gamma^2}	.
\end{align}
Therefore,
\begin{align}
\left| u^\eps - \left( u^{(0)} + \sqrt{\eps} u^{(1)}\right) \right|
	&= 			\left| 
					\left(u^{(0)} + \sqrt{\eps} \, u^{(1)} + \eps \, u^{(2)} + \eps \sqrt{\eps}\,u^{(3)} + R^\eps \right) - \left( u^{(0)} + \sqrt{\eps} u^{(1)}\right)
					\right|	\\
	&\leq		\left| R^\eps \right| +	\eps \left| u^{(2)} + \sqrt{\eps}\,u^{(3)} \right|	\\
	&\leq		\eps \, C_1 \, T e^{T C_\Gamma^2}	+ \eps \, C_2	\\
	&\leq		\eps \,	C \, T e^{T C_\Gamma^2}	,
\end{align}
for some constants $C_2$ and $C$.  This establishes Theorem \ref{thm:uBound}.
\end{proof}
\noindent Now we state our main accuracy result.
\begin{theorem}\label{thm:PBound}
For fixed $(t,x,y,z)$ there exists a constant $C$ such that for any $\eps<1$ the price of a European option $P^\eps(t,x,y,z)$ given by \eqref{eq:Ptemp}, satisfies
\begin{align}
\left| P^\eps - \left(P^{(0)}+\sqrt{\eps}\,P^{(1)}\right) \right| &\leq \eps \, C .
\end{align}
\end{theorem}
\begin{proof}
\begin{align}
&	\left| P^\eps(t,x,y,z) - \left( P^{(0)}(t,x,z) + \sqrt{\eps}\, P^{(1)}(t,x,z)\right) \right|	\\
&\qquad		=			\left| e^{-r t} \Etilde_{x,y,z} \left[ u^\eps(T_t,x,y,z;t) - \left( u^{(0)}(T_t,x,z;t) + \sqrt{\eps}\, u^{(1)}(T_t,x,z;t)\right) \right] \right|	\\
&\qquad		=	e^{-rt} \left| \Etilde_{x,y,z} \left[ \Etilde_{x,y,z} \left[ 
								u^\eps(T_t,x,y,z;t) - \left( u^{(0)}(T_t,x,z;t) + \sqrt{\eps}\, u^{(1)}(T_t,x,z;t)\right) \Big| \, T_t 
								\right] \right] \right|	\\
&\qquad		\leq	\Etilde_{x,y,z} \left[ \Etilde_{x,y,z} \left[
								\left| u^\eps(T_t,x,y,z;t) - \left( u^{(0)}(T_t,x,z;t) + \sqrt{\eps}\, u^{(1)}(T_t,x,z;t)\right) \right| \Big| \, T_t 
								\right] \right]	\\
&\qquad		\leq	\Etilde_{x,y,z} \left[
								\eps \,	C_1 \, T_t e^{T_t C_\Gamma^2} 
								\right]	\qquad \text{(by Theorem \ref{thm:uBound})}\\
&\qquad		\leq	\eps \, C_1 \, C_T \hspace{27 mm} \text{(by assumption \ref{assume:T} of section \ref{sec:assumptions})}	\\
&\qquad		=			\eps \, C	,
\end{align}
for some constants $C_1$ and $C$.  This establishes Theorem \ref{thm:PBound}.
\end{proof}

\section{Call Option Examples}\label{sec:examples}
In this section we provide examples of how to calculate the price of a European call option in four different time-change regimes.  These examples demonstrate both the flexibility and analytic tractability of the TC-FMR-SV framework.
\subsection{FMR-SV} \label{sec:FMRSVex}
The first regime we consider is that of no random time-change (i.e. $T_t = t$).  This choice for $T_t$  reduces the TC-FMR-SV framework to that of pure FMR-SV.  To calculate the approximate price of a European call option $P^{(0)}(t,x) + \sqrt{\eps} \, P^{(1)}(t,x)$ we use equations \eqref{eq:P0} and \eqref{eq:P1} of Theorem \ref{thm:main}.  Since $T_t$ is not random in the present scenario, expectations \eqref{eq:Expectation0} and \eqref{eq:Expectation1} reduce to
\begin{align}
\Etilde_z \left[ e^{-\Lam_\om^{(0)} T_t}\right]																										&=	e^{-\Lam_\om^{(0)} t}	,		&
\Etilde_z \left[ \l( - \Lam_\om^{(1)} T_t \r)e^{-\Lam_\om^{(0)} T_t}\right]						&=	\l( - \Lam_\om^{(1)} t \r) e^{-\Lam_\om^{(0)} t}	.
\end{align}
The $\O\l(\eps^0\r)$ eigenfunctions $\Psi_\om^{(0)}(x)$ are given in \eqref{eq:Psi0}.  Expressions for $\Lam_\om^{(0)}$ and $\Lam_\om^{(1)}$ can be found in \eqref{eq:Lambda0} and \eqref{eq:Lambda1} respectively.  Hence, what remains in order to calculate the approximate price of a call option $P^{(0)}(t,x) + \sqrt{\eps} \, P^{(1)}(t,x)$ is an expression for $C_\om^{(0)}(t)$.
\par
For a European call with strike price $K=e^k$ and time of maturity $t<\infty$, the option payoff $h(S_t)$ is given by
\begin{align}
h(S_t)	&=	\left(S_t - e^k\right)^{+}	.
\end{align}
Using equation \eqref{eq:C0} we calculate
\begin{align}
C_\om^{(0)}(t)	&=		\left( \Psi_\omega^{(0)}(\cdot), h\left(e^{rt\,+\,\cdot \, }\right)\right)_s	\\
											&=		\int_\R 	\frac{1}{\sqrt{2\pi}} e^{\l(-i\omega+1/2\r) x} \left(e^{rt+x}-e^k\right)^+ e^{-x} dx	\label{eq:integral}		\\
											&=		\frac{e^{k (1/2 - i \omega )} \left(4 i \omega-2 \right) - e^{r t + k (1/2 - i \omega )} \left(4 i \omega + 2 \right)}{\sqrt{2 \pi } \left(1+4 \omega ^2\right)} .
											\label{eq:C0call}
\end{align}
Note that integral \eqref{eq:integral} will not converge for purely real values of $\om$.  However if we  move $\om$ into the complex plane $\om = \om_r + i \om_i$ and we fix the imaginary part of $\om$ such that $\om_i<(-1/2)$, then integral \eqref{eq:integral} will converge.  Upon doing this, when calculating option prices using \eqref{eq:P0} and \eqref{eq:P1}, we must remember to hold the imaginary part $\om_i<(-1/2)$ fixed and integrate with respect to the real part of $\om$ (i.e. set $d\om = d\om_r$).
\par
Figure \ref{fig:ImpVolFPS} demonstrates the implied volatility surface induced by the TC-FMR-SV framework in the $T_t = t$ regime.  We plot implied volatilities $I$ versus $\log$-moneyness-to-maturity ratio (LMMR).  We remind the reader that $I$ and LMMR are defined by
\begin{align}
P^{BS}(t,K,I)		&=		\l(P^{(0)} + \sqrt{\eps}P^{(1)}\r)(t,K)	,	\\
\text{LMMR}		&=		\log \l( K / S_0 \r) / t	,
\end{align}
where $P^{BS}(t,K,I)	$ is the Black-Scholes price of a call option with strike price $K$, time to maturity $t$ and volatility $I$.  The notation $\l(P^{(0)} + \sqrt{\eps}P^{(1)}\r)(t,K)$ is used here to indicate the approximate price of call option as calculated in the TC-FMR-SV framework with strike price $K$ and time to maturity $t$.  The parameters used in figure \ref{fig:ImpVolFPS} are
\begin{align}
r=0.00, \sig=0.34, V_2^\eps = 0.03, V_3^\eps = -0.03	.
\end{align}
We note that the volatility surface induced by the $T_t = t$ regime is able to produce a negative at-the-money (ATM) skew, which is typical of equity call options.  However, without a stochastic time-change, the implied volatility surface will not exhibit a smile \cite{fouque}.

\subsection{TC-FMR-SV: L\'{e}vy Subordinator}
Next, we consider a regime where the random time-change $T_t^1$ is given by a L\'{e}vy subordinator.  The jumps of our prototype L\'{e}vy subordinator will be modeled as a compound Poisson process.  Specifically, we consider
\begin{align}
T_t^1		:=		\gamma t + \sum_{i=1}^{N_t^\alpha} \xi_i	,			\label{eq:T1exp}
\end{align}
where $\gamma$ is the drift of the L\'{e}vy subordinator, $N_t^\alpha$ is a homogeneous Poisson process with jump-arrival intensity $\alpha$ and the $\xi_i$ are i.i.d. random variables with exponential distribution $\xi_i \sim {\cal E}(\eta)$ and mean $\Etilde[\xi_i]=1/\eta$.  As noted in section \ref{sec:Levy} the L\'{e}vy measure $\nu(ds)$ of a compound Poisson process can be written as the product of the (net) jump-arrival intensity $\alpha$ and the distribution $F_\xi(s)$ of the i.i.d. jumps.  In this case
\begin{align}
\nu (ds) 	&=	\alpha F_\xi (ds)	,	&
F_\xi(s)	&=	1 - e^{-\eta s}	.	\label{eq:Fexp}
\end{align}
Using equations \eqref{eq:LevyKintchine2} and \eqref{eq:Fexp}, we calculate the L\'{e}vy exponent $\phi(\Lambda)$ of a $T_t^1$ as
\begin{align}
\phi(\Lambda)	&=	\gamma \Lambda + \frac{\alpha \Lambda}{\Lambda + \eta}	,	&
											&\l( \Lam > -\eta \r)	.
\end{align}
For a L\'{e}vy subordinator expectations \eqref{eq:Expectation0} and \eqref{eq:Expectation1} reduce to
\begin{align}
\Etilde_z \left[ e^{-\Lam_\om^{(0)} T_t^1}\right]																								&=	e^{-\phi_\om^{(0)} t}	,	&
\Etilde_z \left[ \l( - \Lam_\om^{(1)} T_t^1 \r)e^{-\Lam_\om^{(0)} T_t^1}\right]						&=	\l( -\phi_\om^{(1)} t\r) e^{-\phi_\om^{(0)} t}	,
\end{align}
where, for a compound Poisson process with exponentially distributed jumps, $\phi_\om^{(0)}$ and $\phi_\om^{(1)}$ are given by
\begin{align}
\phi_\om^{(0)}					&=		\gamma \Lambda_\om^{(0)}+\frac{\alpha  \Lambda_\om^{(0)}}{\eta +\Lambda_\om^{(0)}}	,	&
\phi_\om^{(1)}					&=		\gamma  \Lambda_\om^{(1)}
															- \frac{\alpha  \Lambda_\om^{(0)} \Lambda_\om^{(1)}}{(\eta +\Lambda_\om^{(0)})^2}
															+ \frac{\alpha  \Lambda_\om^{(1)}}{\eta +\Lambda_\om^{(0)}}	.
\label{eq:phiExp}
\end{align}
The coefficients $C_\om^{(0)}(t)$, given by \eqref{eq:C0call}, are unaffected by the choice of random time-change.  Hence, the approximate price of a European call option $P^{(0)}(t,x) + \sqrt{\eps} \, P^{(1)}(t,x)$ can now be calculated using \eqref{eq:P0} and \eqref{eq:P1}.
\par
Figure \ref{fig:ImpVolLevy} plots implied volatilities $I$ versus LMMR in the TC-FMR-SV regime in which $T_t^1$ is given by \eqref{eq:T1exp}.  The parameters used in figure \ref{fig:ImpVolLevy} are
\begin{align}
r=0.00, \sig=0.34, V_2^\eps = 0.03, V_3^\eps = -0.03, \eta=0.10, \alpha=0.75, \gamma=0.25 
\end{align}
We note that the implied volatility surface of figure \ref{fig:ImpVolLevy} exhibits an ATM skew as well as a true smile with implied volatilities rising at the largest strikes.  The strong skew and smile are particularly noticeable at shorter maturities.  This is consistent with the findings of \cite{gatheral2}, where it was noticed that a model for the underlying asset $S_t$ must contain jumps in order for the induced implied volatility surface to capture the steep skew and strong smile of the empirically observed implied volatility surface for short-maturity options.

\subsection{TC-FMR-SV: Absolutely Continuous Time-Change}
Recall that an absolutely continuous time-change $T_t^2$ is of the form \eqref{eq:T2}.  As an example, we consider $Z_t$ to be the classic Cox–-Ingersoll–-Ross (CIR) process and the rate function to be the identity $V(z)=z$.  We have
\begin{align}
dZ_t 						&= 	\kappa (\Theta-Z_t) dt + \Sigma \sqrt{Z_t} d\Wtilde^z_t	,	\label{eq:CIR} \\
T_t^2						&=	\int_0^t Z_s \, ds	,	\label{eq:T2cir}
\end{align}
where $\Wtilde^z_t$ is a Brownian motion under $\Ptilde$.  Here $\kappa>0$ is the rate of mean-reversion of the CIR process and $\Theta>0$ is the long-run mean. We shall refer to $\Sigma>0$ as the ``vol of vol'' since $\Sigma$ controls the volatility of $Z_t$, which in turn contributes to the volatility of $X_{T_t^2}$.  We shall enforce the condition $2 \kappa \Theta \geq \Sigma^2$ so that the CIR process $Z_t$ remains strictly positive for all time (see \cite{lamberton}, Chapter 6).
\par
In order to compute option prices in the absolutely continuous time-change regime, we need to know the Laplace transform of $T_t^2$.  This is a classical calculation, which can be found in \cite{lamberton}.  Here, we simply state the result  
\begin{align}
L(t,z,\Lambda)	&=	\Etilde_z \left[e^{-\Lambda T_t^2} \right]	\\
								&=	e^{-\kappa \Theta U(t) - z V(t)}	,	&
								&\l( \Lam \geq \frac{-\kappa^2}{2 \Sigma^2} \r) ,	\\
U(t)						&=	\frac{-2}{\Sigma ^2} \log\left[
										\frac{2 \gamma e^{(\gamma+\kappa) t/2 }}{(\gamma - \kappa )+e^{\gamma t } (\gamma +\kappa )}
										\right]	,	\\
V(t)						&=	\frac{2\Lambda  \left(e^{\gamma  t }-1\right) }{(\gamma -\kappa )+e^{\gamma  t } (\gamma +\kappa )}	,	\\
\gamma					&=	\sqrt{\kappa^2 + 2 \, \Sigma^2 \Lambda }	.
\end{align}
For an absolutely continuous time-change $T_t^2$, expectations \eqref{eq:Expectation0} and \eqref{eq:Expectation1} reduce to
\begin{align}
\Etilde_z \left[ e^{-\Lam_\om^{(0)} T_t^2}\right]																										&=	L\l(t,z,\Lam_\om^{(0)}\r) , &
\Etilde_z \left[ \l( - \Lam_\om^{(1)} T_t^2 \r)e^{-\Lam_\om^{(0)} T_t^2}\right]						&=	\d_\alpha L\l(t,z,\Lam_\om^{(1)}\alpha\r)	\Big|_{\alpha=\Lam_\om^{(0)}/\Lam_\om^{(1)}} 
\label{eq:ExpectationsCIR}
\end{align}
The above expectations, along with expression \eqref{eq:C0call} for $C_\om^{(0)}(t)$, are enough to calculate the approximate price of a call option $P^{(0)}(t,x,z) + \sqrt{\eps} \, P^{(1)}(t,x,z)$ using \eqref{eq:P0} and \eqref{eq:P1}.
\par
Figure \ref{fig:ImpVolAbsCont} plots implied volatilities $I$ versus LMMR in the TC-FMR-SV regime in which $T_t^2$ is given by \eqref{eq:T2cir}.  The parameters used in figure \ref{fig:ImpVolAbsCont} are
\begin{align}
r=0.00, \sig=0.34, V_2^\eps = 0.03, V_3^\eps = -0.03, \kappa = 1.00, \Theta = 1.00, \Sigma^2 = 2.00, z= 2.00	.
\end{align}
We observe that the implied volatility surface in figure \ref{fig:ImpVolAbsCont} exhibits an ATM skew as well as a slight smile effect.  Though, neither the skew nor smile in figure \ref{fig:ImpVolAbsCont} is as pronounced as in figure \ref{fig:ImpVolLevy} where the stochastic time-change is given by a L\'{e}vy subordinator $T_t^1$.

\subsection{TC-FMR-SV: Composite Time-Change}
Finally, we consider a composite time-change $T_t^3 = T_{T_t^2}^1$ where $T_t^1$ is the L\'{e}vy subordinator described by equation \eqref{eq:T1exp} and $T_t^2$ is the absolutely continuous time-change described by equation \eqref{eq:T2cir}.  In this regime expectations \eqref{eq:Expectation0} and \eqref{eq:Expectation1} can be found by replacing $\Lam_\om^{(0)}$ and $\Lam_\om^{(1)}$ in \eqref{eq:ExpectationsCIR} by $\phi_\om^{(0)}$ and $\phi_\om^{(1)}$ from equation \eqref{eq:phiExp}.
\par
In figure \ref{fig:ImpVolComp} we plot implied volatility $I$ induced by the composite time-change $T_t^3$ as a function of LMMR.  The parameters used in figure \ref{fig:ImpVolComp} are
\begin{align}
r=0.00, \sig=0.34, &V_2^\eps = 0.03, V_3^\eps = -0.03, \gamma = 0.05, \alpha = 0.50, \eta = 0.50, \\
 \kappa &= 2.00, \Theta = 1.00, \Sigma^2 = 4.00, z= 4.00	.
\end{align}
Once again, we observe an ATM skew and strong smile at the shortest maturity, with these features diminishing for longer maturities.

\section{Summary and Conclusions}\label{sec:conclude}
In this paper we introduce a class of TC-FMR-SV models.  The key features of our modeling framework are:
\begin{enumerate}
	\item  We are able to include jumps in the price process of the underlying asset.
	\item  We can incorporate multiple factors of stochastic volatility, which run on different time scales.
	\item  We are able to account for the empirically observed negative correlation between asset returns and volatility (the leverage effect).
\end{enumerate}
Some of the main results of our analysis are:
\begin{enumerate}
	\item  We provide simple formulas to calculate the approximate price of any European option.
	\item  By combining different time-changes, we are able to produce a wide array of implied volatility surfaces.
\end{enumerate}
Overall, we feel that the flexibility provided by the TC-FMR-SV framework and the analytic tractability it provides, merit continued research in this area.  A logical next step, for example, would be to incorporate default of the underlying asset into our class of models, as done in \cite{carr}.  Additionally, characterization of the implied volatility surface through an expansion $I^\eps \approx I^{(0)} + \sqrt{\eps}\, I^{(1)}$ would be useful.

\subsection*{Thanks}
The authors of this paper would like to thank Jean-Pierre Fouque and two anonymous referees for their thoughtful comments on this work.  Their suggestions have greatly improved both the quality and readability of this paper.

\appendix

\section{Poisson Equations and the Fredholm Alternative}\label{sec:Poisson}
The purpose of this appendix is to explain why centering condition \eqref{eq:centering} is necessary in order for the Poisson equation \eqref{eq:PoissonGeneric} to admit a solution.  To begin, we consider an ergodic Markov diffusion $Y_t$ that lives on $\R$, has invariant distribution $F_Y(dy) = \rho(y) \,dy$ and whose infinitesimal generator and adjoint are given by
\begin{align}
\L_Y				&=		\mu(y) \d_y + \frac{\sig^2(y)}{2}\d^2_{yy}	,	&
\L_Y^{*}		&=		- \d_y \mu(y) + \, \d^2_{yy} \frac{\sig^2(y)}{2} \, .
\end{align}
From the Kolmogorov forward equation, the density $\rho(y)$ satisfies $\L_Y^{*}\rho = 0$ and is given by
\begin{align}
\rho(y)	&=	C \frac{2}{\sig^2(y)}\exp\l(\int^y \frac{2 \mu(z)}{\sig^2(z)} dz \r)	,
\end{align}
where $C$ is a constant such that $\int \rho(y) \,dy =1$.
\par
We consider the following Poisson problem: find $\Psi \in C^2(\R)$ such that
\begin{align}
\L_Y \Psi + g	 &= 0 && \text{in $\R$},	\label{eq:LPsi=g} \\
\lim_{y \rightarrow \pm \infty} | \Psi(y) | &< |y|^p &&\text{for some real $p<\infty$}.
\end{align}
Multiplying  $\l(\L_Y \Psi(y)\r)$
by $\rho(y)$ and integrating with respect to $y$ we find
\begin{align}
\int_{-\infty}^{\infty} \rho \l( \L_Y \Psi \r) \, dy	&=	 	\Big[ C \exp \l( \int^y \frac{2 \mu(z)}{\sig^2(z)} dz \r) \d_y \Psi(y) \Big|_{-\infty}^\infty + \int_{-\infty}^\infty  \Psi \l( \L_Y^{*} \rho\r) \,dy	\\
																																&=		\Big[ \frac{\sig^2(y)}{2} \rho(y) \d_y \Psi(y) \Big|_{-\infty}^\infty	,
\end{align}
where we have used integration by parts and $\L_Y^{*}\rho = 0$.  Hence,
\begin{align}
\<g\> :=	\int_{-\infty}^{\infty} \rho \, g \, dy 		&=  	- \Big[ \frac{\sig^2(y)}{2} \rho(y) \d_y \Psi(y) \Big|_{-\infty}^\infty.								\label{eq:PoissonBCs}																				
\end{align}
In section \ref{sec:eigenvalue} we considered Poisson equations with respect to the operator $\L^{(-2)} = \L_Y^1$, the infinitesimal generator of the volatility-driving process $Y_t^1$ under the physical measure $\P$, which we chose to be an OU process.  Under the physical measure the OU process $Y_t^1$ has an invariant distribution $F_Y \sim \mathcal{N}(m,\nu^2)$.  In this setting, $\rho(y)$ in equation \eqref{eq:PoissonBCs} asymptotically behaves like $\sim \, e^{-y^2}$.
Thus, if we restrict ourselves to solutions $\Psi(y)$ of \eqref{eq:LPsi=g} that have at most polynomial growth as $y \rightarrow \pm \infty$ then the right-hand side of \eqref{eq:PoissonBCs} is zero.  Hence, a necessary condition for the solvability of \eqref{eq:LPsi=g} becomes $\<g\>=0$, which is precisely the centering condition given in \eqref{eq:centering}.
\par
We have established that $\<g\>=0$ is a \emph{necessary} condition for the solvability of \eqref{eq:LPsi=g}.  It turns out $\<g\>=0$ is also a \emph{sufficient} condition for  \eqref{eq:LPsi=g} to have a solution.  The \emph{Fredholm alternative} states that one of the following is true:
\begin{enumerate}
\item $\L_Y \Psi + g=0$ has a unique solution (i.e. $\L_Y$ is invertible) \emph{or}
\item $\L_Y \Psi =0$ has a non-trivial solution, in which case $\L_Y \Psi + g=0$ has a solution if $g \perp \text{Ker}\l(\L_Y^{*}\r)$.
\end{enumerate}
For the OU process with infinitesimal generator $\L_Y^1$, we note that $\L_Y^1 \Psi = 0$ has a non-trivial solution -- namely $\Psi(y) = 1$.  Hence
by the Fredholm alternative $\L_Y^1 \Psi + g =0$
has a solution if $g \perp \text{Ker}\l(\L_Y^{1*}\r)$.  Since we have $\text{Ker}\l(\L_Y^{1*}\r)=\left\{\rho(y)\right\}$, the statement $g \perp \text{Ker}\l(\L_Y^{1*}\r)$ is equivalent to the centering condition $\<g\>=0$.  The following (formal) solution can easily be checked
\begin{align}
\Psi(y)	&=	\int_0^\infty e^{s \L_Y^1} g(y) \, ds	.
\end{align}
We refer the reader to section $6.6.3$ of \cite{fouque2007wave} for a detailed exposition on Poisson equations and the Fredholm alternative.
%

%
%


\clearpage

\begin{figure}
	\centering
		\includegraphics[width=1.00\textwidth]{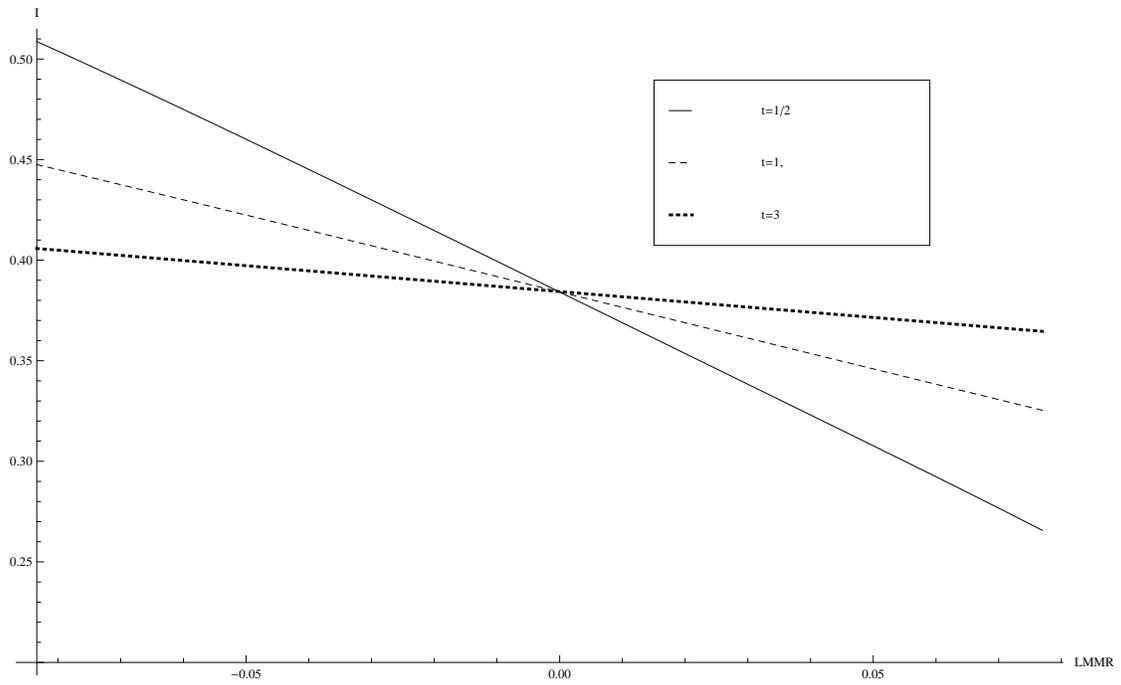}
	\caption{Implied volatility surface induced by FMR-SV.}
	\label{fig:ImpVolFPS}
\end{figure}

\begin{figure}
	\centering
		\includegraphics[width=1.00\textwidth]{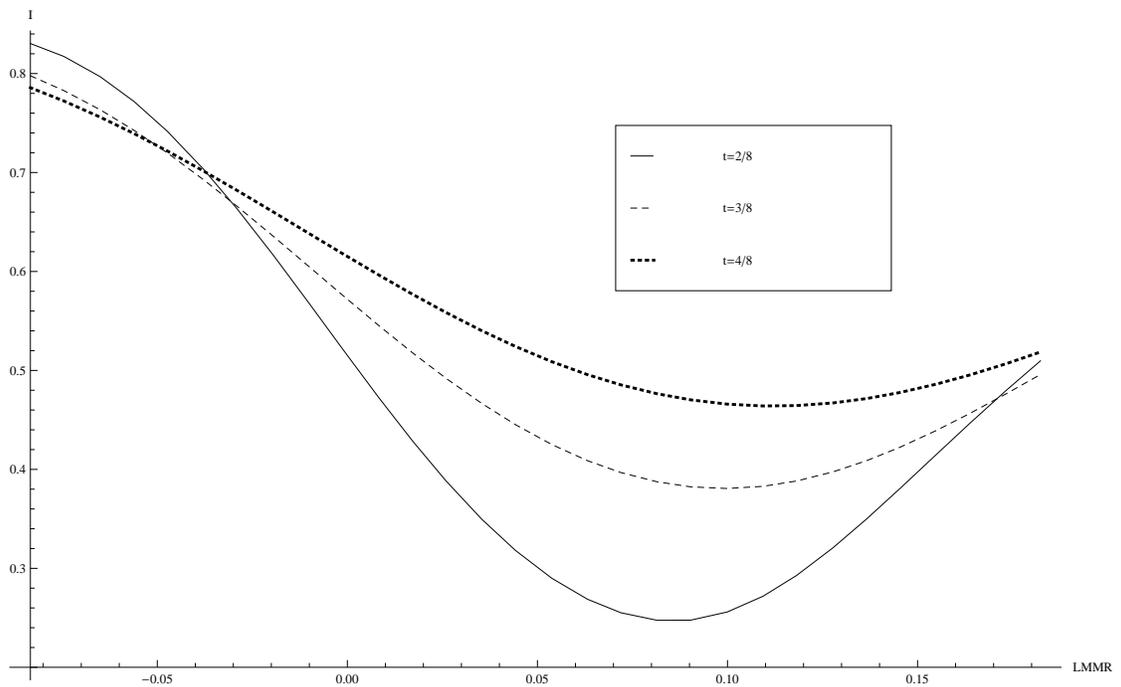}
	\caption{Implied volatility surface induced by a L\'{e}vy subordinator.}
	\label{fig:ImpVolLevy}
\end{figure}

\begin{figure}
	\centering
		\includegraphics[width=1.00\textwidth]{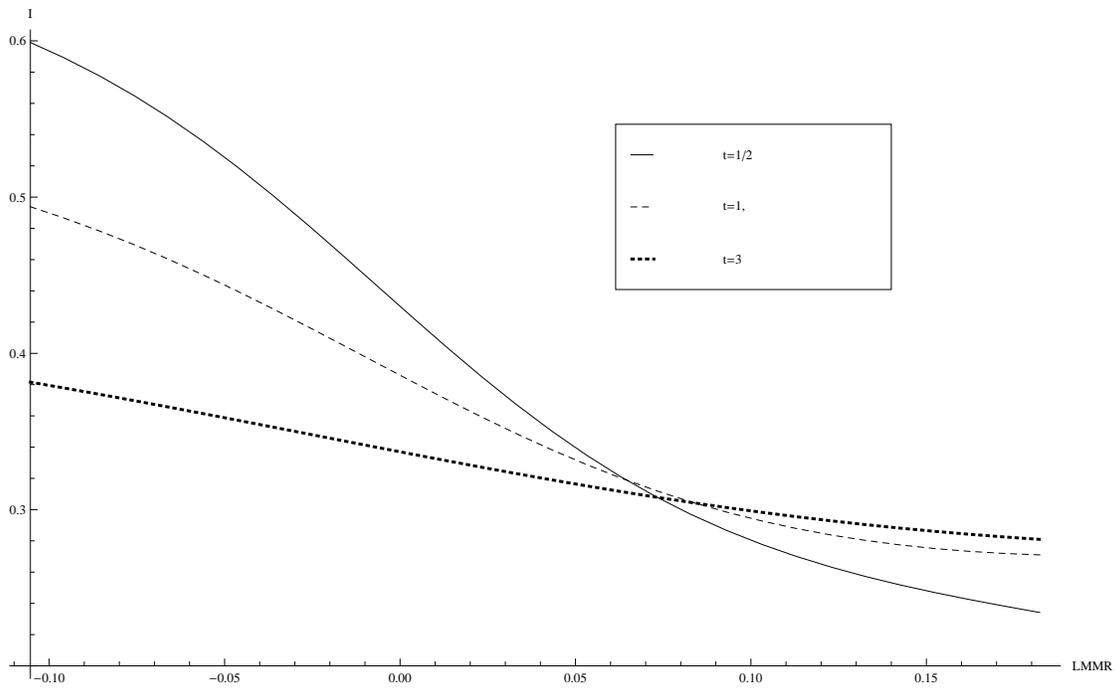}
	\caption{Implied volatility surface induced by an absolutely continuous time-change.}
	\label{fig:ImpVolAbsCont}
\end{figure}

\begin{figure}
	\centering
		\includegraphics[width=1.00\textwidth]{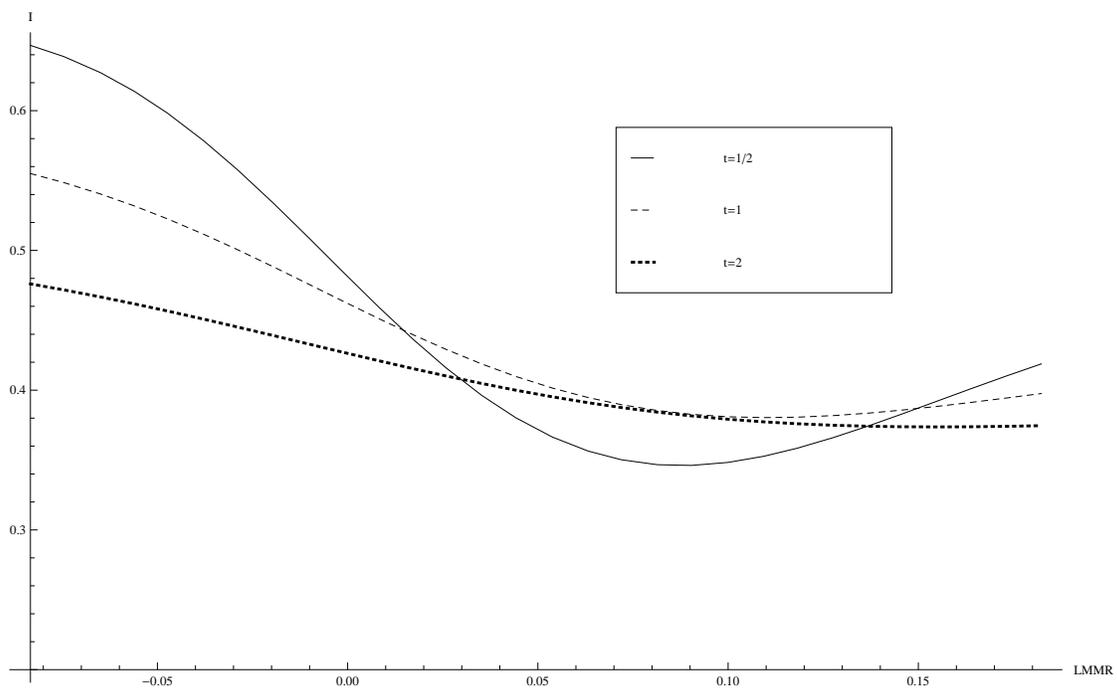}
	\caption{Implied volatility surface induced by a composite time-change.}
	\label{fig:ImpVolComp}
\end{figure}


\clearpage 					
\bibliographystyle{siam}
\bibliography{ThesisBibtex}

\end{document}